\documentclass[journal]{IEEEtran}
\usepackage{bbm}
\usepackage{mathrsfs}
\usepackage{amsfonts}
\usepackage[dvips]{graphicx}
\usepackage{epsfig,latexsym}
\usepackage{float}
\usepackage{indentfirst}
\usepackage{amsmath}
\usepackage{bm}
\usepackage{amssymb}
\usepackage{times}
\usepackage{subfigure}
\usepackage{psfrag}
\usepackage{cite}
\usepackage{subfigure}
\usepackage{algorithm}
\usepackage{algorithmic}
\usepackage{color}
\usepackage{array}

\newtheorem{theorem}{\textbf{Theorem}}
\newtheorem{lemma}{\textbf{Lemma}}

\newtheorem{definition}{\textbf{Definition}}
\newtheorem{proposition}{\textbf{Proposition}}
\newtheorem{remark}{\textbf{Remark}}

\begin{document}

\title{Queue-Aware Energy-Efficient Joint Remote Radio Head Activation and Beamforming in Cloud Radio Access Networks}

\author{Jian~Li,~Jingxian~Wu,~\IEEEmembership{Senior
Member,~IEEE},~Mugen~Peng$^{\dagger}$,~\IEEEmembership{Senior
Member,~IEEE},~Ping~Zhang$^{\dagger}$,~\IEEEmembership{Senior
Member,~IEEE}
\thanks{Jian~Li (e-mail: {\tt lijian.wspn@gmail.com}), Mugen~Peng (e-mail:
{\tt pmg@bupt.edu.cn}), Ping~Zhang (e-mail:pzhang@bupt.edu.cn) are
with the Key Laboratory of Universal Wireless Communications for
Ministry of Education, Beijing University of Posts and
Telecommunications, Beijing, 100876, China. Jingxian~Wu (e-mail:
{\tt wuj@uark.edu}) is with the Department of Electrical
Engineering, University of Arkansas, Fayetteville, 72701, USA.}
\thanks{Part of this work has been presented at IEEE Global Communications Conference (GLOBECOM), San Diego, CA, USA, Dec. 2015.}}

\maketitle

\begin{abstract}
In this paper, we study the stochastic optimization of cloud radio access networks (C-RANs) by
joint remote radio head (RRH) activation and beamforming in the
downlink. Unlike most previous works that only consider a static
optimization framework with full traffic buffers, we formulate a
dynamic optimization problem by explicitly considering the effects
of random traffic arrivals and time-varying channel fading. The
stochastic formulation can quantify the tradeoff between power
consumption and queuing delay. Leveraging on the Lyapunov
optimization technique, the stochastic optimization problem can be
transformed into a per-slot penalized weighted sum rate maximization
problem, which is shown to be non-deterministic polynomial-time
hard. Based on the equivalence between the penalized weighted sum
rate maximization problem and the penalized weighted minimum mean
square error (WMMSE) problem, the group sparse beamforming
optimization based WMMSE algorithm and the relaxed integer
programming based WMMSE algorithm are proposed to efficiently obtain
the joint RRH activation and beamforming policy. Both algorithms can
converge to a stationary solution with low-complexity and can be
implemented in a parallel manner, thus
they are highly scalable to large-scale C-RANs. In addition, these two
proposed algorithms provide a flexible and efficient means to adjust
the power-delay tradeoff on demand.
\end{abstract}


\section{Introduction}

{The fifth-generation (5G) wireless networks are expected to provide
ubiquitous services to a larger number of simultaneous mobile
devices with device density far beyond the current wireless
communication systems. To cope with these challenges, ultra-dense
low power nodes and cloud computing are regarded as two of the most
promising techniques{\cite{chiling}}.} Leveraged on low power node
and cloud computing, the cloud radio access network (C-RAN), first
proposed in{\cite{cmcc}}, is expected to revolutionize the
architecture and operations of future wireless systems, and it has
attracted considerable amount of attentions in both academia and
industry{\cite{5G1}\cite{CranMGMN}}. { As shown in Fig.
\ref{cran_rchitecture}, a large number of remote radio heads (RRHs)
are densely deployed in the space domain for C-RANs. Each RRH is
configured only with the front radio frequency (RF) components and
some basic transmission/reception functionalities. The RRHs are
connected to the baseband unit (BBU) pool through high-bandwidth and
low-latency fronthaul links to enable real-time cloud computing. The
C-RANs can act as a platform for the practical implementation of
coordinated multi-point (CoMP) transmission concepts{\cite{5G}}.
{Specifically, the BBU pool computes the beamforming weight
coefficients for different RRHs, and sends the precoded data to
various active RRHs. Then the active RRHs cooperatively transmit the
precoded data to different UEs.} The signals observed at each UE are
superpositions of signals from multiple active RRHs. The beamforming
weight coefficients are designed to steer the data to their intended
receivers in the spatial domain. That is, for a given UE, the
desired signals are combined coherently yet the interfering signals
are combined out-of-phase. Here the joint beamforming aims to
improve the signal-to-interference-plus-noise ratio (SINR) in order
to significantly improve the spectral efficiency of C-RANs.}

\begin{figure}
\centering
\includegraphics[scale=0.30]{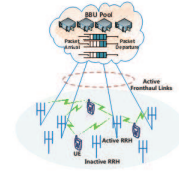}
\caption{Architecture of heterogeneous cloud radio access networks}
\label{cran_rchitecture}
\end{figure}

Various beamforming designs have been studied for CoMP in C-RANs
with different optimization objectives and constraints recently.
{In\cite{quek} and\cite{lau}, the number of active front haul links
is minimized under a SINR constraint for each user equipment (UE)
and a power constraint for each RRH. The problem is solved by
minimizing an approximate of the original combinatorial objective
function. The works in\cite{VNHa} and \cite{BDai} aim to jointly
optimize the set of RRHs serving each user and the corresponding
beamformers, under the constraint of front haul capacity.
Specifically, power minimization is studied in\cite{VNHa} and sum
rate maximization is considered in\cite{BDai}.} The design problems
in\cite{VNHa}\cite{BDai} fundamentally differ from that
in\cite{quek}\cite{lau} in that they explicitly consider the
fronthaul capacity. The problem of minimizing the overall power
consumption and CoMP operational costs by joint user association and
downlink beamforming was studied in\cite{MINP2}, where the problem
was addressed using a mixed integer second-order cone program
framework. Similar problem has been studied in\cite{SLuo}, where the
user association and beamforming were considered in both downlink
and uplink C-RANs. As existing solutions considering only the
downlink like\cite{MINP2} cannot be modified in a straightforward
way to solve the problem in\cite{SLuo}, efficient algorithms were
proposed utilizing the uplink-downlink duality result.

The dense deployment of RRHs imposes new technical challenges for
the design and implementation of large scale C-RANs. With the
centralized processing at the BBU pool, the power consumed by the
fronthaul links that provide high-capacity connections with BBU pool
becomes comparable to that for transmission{\cite{fronthaulpower}}.
Therefore, in order to reduce power consumption of the entire
network, we can reduce the number of active RRHs by putting some of
the RRHs into the sleep mode. The fronthaul links of sleeping RRH
will also be turned off to save power consumption. Therefore, the
scheduling of RRH activation plays a critical role in the
development of energy-efficient C-RANs. Related problem has been
studied in{\cite{groupgreen}}, where two efficient group sparse
beamforming algorithms were proposed to obtain the subset of active
RRHs and their corresponding beamformers. {The works
in\cite{BDaiJSAC} compares the energy efficiencies of two different
downlink transmission strategies in C-RANs by taking the RRH
transmission power, RRH activation power, and load-dependent
fronthaul power into considerations.} Compared with the optimal
exhaustive search method, the computational complexity of the
algorithms
in\cite{quek,lau,VNHa,BDai,MINP2,SLuo,groupgreen,BDaiJSAC} can be
significantly reduced, which, however, can still be very intensive
for large-scale C-RANs. This is due to the fact that a series of
convex problems (e.g. SDP, SOCP) have to be solved centrally using
standard CVX solvers. Furthermore, the aforementioned literatures
are typically based on snapshot-based static models, which indicates
that the stochastic and time-varying features are not considered
into the formulations. Therefore, only the physical layer
performance metrics such as power and throughput are optimized and
the resulting control policy is only adaptive to channel state
information (CSI). In practice, delay is also a key metric to
measure the quality-of-service (QoS), which has also been neglected
in these literatures. Intuitively, there is a fundamental tradeoff
between power consumption and queuing delay{\cite{cuiying}}, thus it
is important to jointly consider power consumption and delay to
balance their tradeoff and to meet various performance requirement
in C-RANs.

In practice, the stochastic control and delay analysis are usually
investigated from the queue stability perspective in a time-varying
system using the Lyapunov optimization technique. Many existing
literatures have focused on stochastic optimization for time-varying
wireless networks. A fundamental approach to stochastic resource
allocation and routing for heterogeneous data networks was presented
in\cite{MNeelyFair}, where the flow control is crucial to ensure no
network resources are wasted whenever the traffic rates are inside
or outside the capacity region. The authors of\cite{LongBLe}
investigated stochastic control for wireless networks with finite
buffers, where the joint flow control, routing, and scheduling
algorithms can achieve high network utility and deterministically
bounded backlogs inside the network. The delay analysis was
conducted in\cite{NeelyDelayAnalysis} for suboptimal scheduling in
one-hop wireless networks with general interference set constraints
and time-correlated traffic arrivals. {There also have been lots of
works that focus on optimizing power under queue stability and
interference constraints\cite{r1}\cite{r2}. However, these works
generally adopt highly simplified physical-layer models for wireless
channels, such as interference avoidance constraint or simple
channel-rate mapping function. They do not consider the complex
non-linear relationship between signal power, interference power,
and system throughput as in practical systems, which makes the
system design much more challenging. This paper fundamentally
differs from previous works in that we consider a power minimization
problem by designing queue-aware joint optimization algorithms for
C-RANs where both RRH activation set and beamforming vector are
adaptive to queue state information and channel state information.
In addition, the non-linear impacts of interence are explicitly
quantified during the system design. Therefore, existing solutions
can not be readily applied to the C-RAN setting considered in this
paper.}

Motivated by these facts, we propose to study dynamic joint RRH
activation and beamforming for C-RANs by considering random traffic
arrivals, queuing delays, and the time-varying fading channels. To
the best of our knowledge, this is the first such work in C-RANs.
Unlike the static optimization problems studied in the literatures,
the problems considered in this paper are formulated as stochastic
optimizations, which are notoriously difficult to solve but are
important for practical systems. The main contributions of this
paper are summarized as follows:

\begin{itemize}
{
\item The stochastic optimization of joint RRH activation and
beamforming is developed for practical C-RANs. A wide range of
detailed operations and constraints, such as beamforming, RRH
activation, time-varying channel, random traffic arrivals, and
fronthaul power consumption, are considered in the formulation. To
the best of our knowledge, this is the first paper that achieves
stochastic control of power and delay by considering realistic
system operations across multiple protocol layers.

\item To tackle the NP-hardness of the optimization problem, two
low-complexity algorithms are proposed using the group sparse
beamforming (GSB) approach and the relaxed integer programming (RIP)
approach, respectively. Both algorithms can be implemented in a
parallel manner with closed-form expressions, thus they are scalable
to large-scale C-RANs.

\item The delay and power performance of the two
proposed algorithms are numerically evaluated. Significant
performance gains are achieved by the proposed algorithms due to the
fact that they are adaptive to the queue state information. The
proposed algorithms can provide a flexible and efficient means to
control the delay-power tradeoff on demand. All these make the
proposed algorithms attractive and useful in practical
applications.}
\end{itemize}

%
%

The remainder of this paper is organized as follows. In Section II,
we introduce the system model and formulate the stochastic
optimization problem. In Section III, the Lyapunov optimization is
introduced and utilized to obtain a penalized weighted sum rate
maximization problem for each slot. The optimization problem is
solved by two efficient algorithms based on GSB approach and RIP
approach in Section IV and Section V, respectively. Numerical
results are presented in Section VI. Finally, we conclude our paper
in Section VII. The acronyms used in this paper are summarized in
Table I.

\begin{table}[ht]
\centering \caption{{Summary of Acronyms}} {
\begin{tabular}{m{1.1cm}|m{5.4cm}}
\hline
\textbf{Acronym} & \textbf{Description}  \\
\hline C-RAN & cloud radio access network \\
BBU & baseband unit \\
RRH & remote radio head \\
MSE & mean square error \\
MMSE & minimum mean square error \\
WMMSE & weighted minimum mean square error\\
GSB & group sparse beamforming\\
RIP & relaxed integer programming \\
BCD & block coordinate descent\\
LASSO & least absolute shrinkage and selection operator\\
ADMM & alternating direction method of multipliers\\
FJP & full joint processing \\
\hline
\end{tabular}}
\end{table}

Throughout this paper, lower-case bold letters denote vectors and
and upper-case bold letters denote matrices. $\bf{I}$ denotes
identity matrix. $\mathbb{C}$ denotes complex domain and the complex
Gaussian distribution with mean ${\bf m}$ and covariance matrix ${\bf R}$
is represented by $\mathcal {CN}({\bf m} , {\bf R}
)$. $\mathbb{E}[\cdot ]$ and $\rm{det}(\cdot)$ represent expectation
and determinant operators, respectively. ${\rm{Re}}\{\cdot \}$
is the real part operator. $||\cdot||_p$ denotes
$\ell_p$-norm of a vector. The inverse, transpose, conjugate
transpose operators are denoted as $(\cdot)^{-1}$, $(\cdot)^T$, $(\cdot)^H$,
respectively.

\section{system model}


\subsection{Scenario Description}

We consider a downlink C-RAN with $K$ RRHs and $I$ UEs, where each
RRH is equipped with $M$ antennas and each UE has
$N$ antennas. Let $\mathcal {K}$ and $ \mathcal {I}$ denote the set
of RRHs and the set of UEs, respectively. The bandwidth of the system is
$W$. We also assume that the network operates in slotted time with
time dimension partitioned into decision slots indexed by $t \in \{
0,1,2,...\}$ with slot duration $\tau$.

Let ${\bf{H}}_{ki}(t) \in {\mathbb{C}^{N \times M}}$ denote the
channel state information (CSI) matrix from RRH $k$ to UE $i$ at
slot $t$, let ${{\bf{H}}_i}(t) =
{[{\bf{H}}_{1i}(t),{\bf{H}}_{2i}(t),...,{\bf{H}}_{Ki}(t)]} \in
{\mathbb{C}^{N \times MK}}$ denote the CSI matrix from all RRHs to
UE $i$ at slot $t$, and let ${\bf{H}}(t) = [{{\bf{H}}_i}(t), ...,
{{\bf{H}}_I}(t)] \in \mathbb{C}^{N\times MIK}$ denote the network CSI at slot $t$. The channel
is assumed to follow quasi-static block fading, where
each element of ${\bf{H}}(t)$ keeps constant for the duration of a
slot, but is identically and independently distributed (i.i.d.) across
different slots. Let ${\bf{w}}_{ki}(t) \in {\mathbb{C}}^{M \times
1}$ denote the beamforming vector at RRH $k$ for UE $i$ at slot $t$,
let ${\bf{w}}_i(t) = [{\bf{w}}_{1i}^T(t),...,{\bf{w}}_{Ki}^T(t)]^T
\in {\mathbb{C}^{{MK} \times 1}}$ denote the aggregated beamformer
for UE $i$ at slot $t$, let ${{\bf{\tilde w}}}_k(t) =
[{\bf{w}}_{k1}^T(t),...,{\bf{w}}_{kI}^T(t)]^T \in {\mathbb{C}^{MI
\times 1}}$ denote the aggregated beamformer used by RRH $k$ at slot
$t$, and let ${\bf{w}}(t) = [{{\bf{\tilde w}}}_1^T(t),...,
{{\bf{\tilde w}}}_K^T(t)]^T \in {\cal C}^{MIK \times 1}$ denote the aggregated beamformer of the
entire network at slot $t$.

Assume that each UE has its own data stream. Let $a_i(t)$ denote the
data message for UE $i$ at slot $t$. Without loss of generality, we
further assume that $\mathbb{E}[a_i^2(t)] = 1$ and each $a_i(t)$ is
i.i.d. among UEs. With linear beamforming operated centrally in the
BBU pool, the baseband signal to be transmitted by RRH $k$ at slot
$t$ is
\begin{equation}
{{\bf{x}}_k}(t) = \sum\limits_{i \in \mathcal {I}}
{{{\bf{w}}_{ki}(t)}{a_i(t)}}.
\end{equation}

The encoded baseband signal ${{\bf{x}}_k}(t)$ is delivered to RRH
$k$ for radio transmission through corresponding fronthaul link. {It
is worth noting that as we focus on the issue of power-delay
tradeoff in this paper, we assume that the fronthaul links are
provisioned with sufficiently high capacity and negligible
latency.}\footnote{{The impact of finite fronthaul capacity on
fronthaul compression and quantization in C-RANs has been
investigated recently in\cite{park}, which does not consider the
power-delay tradeoff. It is expected that finite fronthaul capacity
will negatively affect the performance of the proposed algorithm.}}
The signal observed by each UE is the superposition of signals from
all RRHs. The received signal at UE $i$ is given by
\begin{equation}
{{\bf{r}}_i}(t) = {{\bf{H}}_i}(t){{\bf{w}}_i}(t){a_i}(t) +
\sum\limits_{j \ne i} {{{\bf{H}}_i}(t){\bf{w}}_j{{(t)}}{{a}_j}(t)} +
{{\bf{z}}_i}(t),
\end{equation}
where ${{\bf{z}}_i}(t) \in {\mathbb{C}^{N \times 1}}$ is the
additive white Gaussian noise (AWGN) at slot $t$ with distribution
$\mathcal {CN}(0, \sigma ^2{{\bf{I}}})$. We assume that all the UEs
adopt single user detection and the interference is treated as
noise. The achievable data rate in the unit of bps/Hz of UE $i$ is
given by
\begin{equation}
\begin{array}{l}
{R_i(t)} = \log_2 \det ({{\bf{I}}} +
{{\bf{H}}_i(t)}{{\bf{w}}_i(t)}{\bf{w}}_i^H(t){\bf{H}}_i^H(t)\\
~~~~~~~~~~(\sum\limits_{j \ne i} {{{\bf{H}}_i(t)}{{\bf{w}}_j(t)}{\bf
w}_j^H(t){\bf{H}}_i^H(t) + {\sigma ^2}} {{\bf{I}}})^{-1}).
\end{array}
\end{equation}

\subsection{Network Power Consumption Model}

In C-RANs, the extensive use of high-capacity low-latency fronthaul
links makes the fronthaul power consumption comparable to the
transmission power of RRHs{\cite{powermodel}}. Here we consider the
passive optical network to provide the effective high-capacity
fronthaul connections between the RRHs and the BBU pool. The passive
optical network consists of optical network units and an optical
line terminal that connects a set of associated optical network
units through a single optical fiber{\cite{PON}}. From the
perspective of energy saving, some RRHs and their associated optical
network units can be switched into sleep mode with negligible power
consumption, but the optical line terminal with constant power
consumption $P_{\rm{OLT}}$ cannot go into sleep mode as it plays the
roles of distributer, aggregator and arbitrator of the transport
network. Here we ignore $P_{\rm{OLT}}$ because it is a constant and
will not affect the scheduling and optimization results. Let
$P_{k}^{\rm{ONU}}$ denote the constant power consumed by the optical
network unit associated with active RRH $k$. Besides, due to the
real-time A/D and D/A processing at each RRH, static circuit power
$P_{k}^{\rm{s}}$ is also consumed. Thus, the amount of static power
consumption associated with RRH $k$ during active mode is
${P_{k}^{\rm{c}}} = {P_{k}^{\rm{s}}} + {P_{k}^{\rm{ONU}}}$. When RRH
$k$ and its corresponding fronthaul link are switched into the sleep
mode, there is no static power consumption. Let $\mathcal {A}(t)
\subseteq \mathcal {K}$ denote the set of active RRHs at slot $t$.
The network power consumption at slot $t$ is given by
\begin{equation}
p(\mathcal {A}(t), {\bf{w}}(t)) = \sum\limits_{k \in \mathcal
{A}(t)} {\left( {\frac{1}{{{\eta _k}}}||{{{\bf{\tilde
w}}}_k}(t)||_2^2 + {P_k^c} } \right)},\label{networkpower}
\end{equation}
where ${\eta _k}$ is the drain efficiency of RF power amplifier at
RRH $k$. { Note that the load-dependent fronthaul power consumption
model has been considered in\cite{BDaiJSAC}, while its impact on our
formulation will be left for future study.} The network power
consumption is a random process, in that it depends on the policy of
RRH activation set and corresponding beamforming vectors, which is
dynamically determined with the observation of traffic queues and
channel conditions at each slot.

\subsection{Queue Stability and Problem Formulation}

The BBU pool maintains $I$ traffic queues for the random traffic
arrivals towards $I$ UEs. Let ${\bf{A}}(t) = [ {A_1}(t)
,...,{A_I}(t)]$ be the vector of stochastic traffic data arrivals
(bits) at the end of slot $t$. We assume that the traffic arrival
${A_i}(t)$ is independent w.r.t. $i$ and i.i.d. over slots according
to a general distribution with mean $\mathbb{E}[{A_i}(t)] = {\lambda
_i}$. Let ${\bf{Q}}(t) = [Q_1(t),...,Q_I(t)]$ denote the vector of
queue state information (QSI) (bits) for the $I$ UEs at the
beginning of slot $t$. Therefore, the queue dynamic for UE $i$ is
given by
\begin{equation}
{Q_i}(t + 1) = \max [{Q_i}(t) - \mu _i(t),0] +
{A_i}(t),\label{qdynamic}
\end{equation}
where the amount of traffic departure at slot $t$ is given by $\mu
_i(t) = W\tau{R_i(t)}$.

To model the impacts of joint RRH activation and beamforming policy
on average queue delay and average network power consumption, we
first present the definitions of queue stability, stability region
and throughput optimal policy as follows{\cite{neely}}.
\begin{definition}[Queue Stability]
\emph{A discrete time queue $Q(t)$ is strongly stable if
\begin{equation}
\mathop {\lim \sup }\limits_{T \to \infty }
\frac{1}{T}\sum\limits_{t = 0}^{T - 1} {\mathbb{E}[Q(t)]} < \infty.
\end{equation}
Furthermore, a network of queues is stable if all individual queues
of the network are stable.}
\end{definition}

\begin{definition}(\emph{Stability Region and Throughput-Optimal Policy}):
\emph{The stability region $\mathcal {C}$ is the closure of the set
of all the arrival rate vectors ${\bm{\lambda}} = \{\lambda _i : i
\in \mathcal {I}\}$ that can be stabilized in a C-RAN. A
throughput-optimal resource optimization policy is a policy that
stabilizes all the arrival rate vectors $\{\lambda _i : i \in
\mathcal {I}\}$ within the stability region $\mathcal {C}$.}
\end{definition}

The objective is to simultaneously maintain the network queue
stability and minimize the network power consumption, by using joint
RRH activation and beamforming. The problem can be formulated as the
following stochastic optimization problem:
\begin{equation}
\begin{array}{ll}
 \min. & \bar {p} = \mathop {\lim }\limits_{T \to \infty } \frac{1}{T}\sum\limits_{t = 0}^{T - 1} {\mathbb{E}[p(\mathcal {A}(t), {\bf{w}}(t))]}\\
 {\rm{s.t.}}& {\rm{C1}:} {\rm{Queue~}}{Q_i}(t){\rm{~is~strongly~stable,~}}\forall i, \\
 & {\rm{C2}:}||{{{\bf{\tilde w}}}_k}||_2^2 \le {P_k}, \\
 \end{array}\label{stochasticproblem}
\end{equation}
{where the expectation $\mathbb{E}$ is taken with respect to the
distribution of network power consumption, which depends on the
random RRH activation set and beamforming vectors.} C1 is the
network stability constraint to guarantee a finite queue length for
each queue. C2 is the constraint on the instantaneous per-RRH power
consumption. In practical C-RANs, the random traffic arrivals and
the time-varying channel conditions are generally unpredictable. The
stochastic nature of the channel conditions and traffic arrivals
makes it impractical to calculate the optimal solution in an offline
manner. To address this problem, we will resort to Lyapunov
optimization, which can transform the stochastic optimization
problem (\ref{stochasticproblem}) into a deterministic one at each
slot.

\begin{remark}
The queue stability constraint is used to depict and control the
average delay. According to Definition 1, the queue stability is
guaranteed if the average queue length is finite. Note that average
delay is proportional to average queue length for a given traffic
arrival rate from Little's Theorem. {As suggested later in Section
III, the average queue length can be arbitrarily bounded by choosing
an appropriate control parameter.}
\end{remark}

\section{Problem Transformation Based on Lyapunov Optimization}

In this section, we will exploit the framework of Lyapunov
optimization to solve the stochastic optimization problem in
(\ref{stochasticproblem}). Define the quadratic Lyapunov function as
$L({\bf{Q}}(t)) = \frac{1}{2}\sum\limits_{i \in \mathcal {I}}
{{Q_i}{{(t)}^2}}$, which serves as a scalar metric of queue
congestion in the C-RAN. To keep the system stable by persistently
pushing the Lyapunov function towards a lower congestion state, the
one-step conditional Lyapunov drift is defined as
\begin{equation}
\Delta ({\bf{Q}}(t)) = \mathbb{E}[ L({\bf{Q}}(t + 1)) -
L({\bf{Q}}(t))|{\bf{Q}}(t)],
\end{equation}
{where $\mathbb{E}$ is the conditional expectation taken with
respect to the distribution of Lyapunov drift given queue state
${\bf{Q}}(t)$.} The Lyapunov drift-plus-penalty function is defined
as
\begin{equation}
\Delta ({\bf{Q}}(t)) + V\mathbb{E}[ p(\mathcal {A} (t),
{\bf{w}}(t))|{\bf{Q}}(t)],\label{driftpenalty}
\end{equation}
{where $\mathbb{E}$ is the conditional expectation taken with
respect to the distribution of network power consumption given queue
state ${\bf{Q}}(t)$,} and $V > 0$ represents an arbitrary control
parameter. The parameter $V$ can be used to control the power-delay
tradeoff. A larger $V$ means more emphasis will be put on power
minimization during the optimization. On the other hand, when $V$ is
small, queue stability carries more weight during the optimization.
Suppose that the expectation of the penalty process $ p(\mathcal {A}
(t), {\bf{w}}(t))$ is deterministically bounded by some finite
constant $p_{\min}$, $p_{\max}$, i.e. ${p_{\min }} \le \mathbb{E}[
p(\mathcal {A} (t), {\bf{w}}(t))] \le {p_{\max }}$. Let $p^*$ denote
the theoretical optimal value of (\ref{stochasticproblem}), then the
relationship between the Lyapunov drift-plus-penalty function and
queue stability is established in Theorem 1{\cite{neely}}.

\begin{theorem}[Lyapunov Optimization]
\emph{Suppose there exist positive constants $B$, $\epsilon$ and $V$
such that for all slots $t \in \{ 0,1,2,...\}$ and all possible
values of ${\bf{Q}}(t)$, the Lyapunov drift-plus-penalty function
satisfies:
\begin{equation}
\Delta ({\bf{Q}}(t)) + V\mathbb{E}[p(\mathcal {A}(t),
{\bf{w}}(t))|{\bf{Q}}(t)] \le B + V{p^*} - \epsilon \sum\limits_{i =
1}^I {{Q_i}(t )},
\end{equation}
then all queues ${Q}_i(t)$ are strongly stable. The average queue
length satisfies
\begin{equation}
\mathop {\lim \sup }\limits_{T \to \infty }
\frac{1}{T}\sum\limits_{t = 0}^{T - 1} {\sum\limits_{i = 1}^I
{\mathbb{E}[ {Q_i}(t )] } } \le \frac{{B + V({p^*} - {p_{min
}})}}{\epsilon },
\end{equation}
and the average penalty of power consumption satisfies
\begin{equation}
\mathop {\lim \sup }\limits_{T \to \infty }
\frac{1}{T}\sum\limits_{t = 0}^{T - 1} {\mathbb{E}[ p(\mathcal
{A}(t), {\bf{w}}(t))] } \le {p^*} + \frac{B}{V}.
\end{equation}}
\end{theorem}

\begin{proof}
The proof can follow that for Theorem 4.2 in\cite{neely}.
\end{proof}

{The results in Theorem 1 motivate us to minimize the Lyapunov
drift-plus-penalty in (\ref{driftpenalty}) to achieve the maximum
queue stability region and obtain throughput-optimal policy. Rather
than directly minimize (\ref{driftpenalty}), our policy actually
seeks to minimize the upper bound of (\ref{driftpenalty}), which is
given by the following lemma [Lemma 4.6 of 23]}.

\begin{lemma}[Upper Bound of Lyapunov Drift-plus-penalty]
\emph{Under any control policy, the drift-plus-penalty has the
following upper bound for all $t$, all possible values of
$\bf{Q}(t)$ and all parameters $V > 0$,
\begin{equation}
\begin{array}{l}
\Delta ({\bf{Q}}(t)) \!+\!\! V\mathbb{E}[p(\mathcal {A}(t),
{\bf{w}}(t))|{\bf{Q}}(t)] \le B +
\\V\mathbb{E}[ p(\mathcal {A}(t),
{\bf{w}}(t))|{\bf{Q}}(t)] + \sum\limits_{i \in \mathcal {I}}
{Q_i}(t)\mathbb{E}[{A_i}(t) - \mu _i(t)|{\bf{Q}}(t)],
\end{array}\label{driftpenaltybound}
\end{equation}
{where $B$ is a positive constant and for all slot $t$ satisfies $B
\ge \frac{1}{2}\sum\limits_{i = 1}^I {\mathbb{E}[{{ {A_i^2}(t)}} +
{{{\mu_i^2}(t)}}|{\bf{Q}}(t)]}$.}}
\end{lemma}

\begin{proof}
The proof is in Appendix A.
\end{proof}

{By the principle of \emph{opportunistically minimizing an
expectation}{\cite{neely}}, the policy that minimizes
$\mathbb{E}[f(t)|{\bf{Q}}(t)]$ is the one that minimizes $f(t)$ with
the observation of ${\bf{Q}}(t)$. Besides, neither $\sum\limits_{i
\in \mathcal {I}}Q_i(t)A_i(t)$ nor $B$ in (\ref{driftpenaltybound})
will be affected by the policy at slot $t$. Therefore,} the
optimization problem can be simplified to
\begin{equation}
\mathop {\max. }\limits_{\mathcal {A}(t),{\bf{w}}(t)}\sum\limits_{i
\in \mathcal {I}} {Q_i(t)}{\mu _i(t)} - V p(\mathcal {A}(t),
{\bf{w}}(t)).\label{P1}
\end{equation}

{The following theorem justifies the throughput optimality by
solving problem (\ref{P1}) optimally.

\begin{theorem}
The RRH activation $\mathcal {A}(t)$ and beamforming ${\bf{w}}{(t)}$
given by solving (\ref{P1}) optimally achieves the maximum stability
region $\mathcal {C}$ in C-RANs.
\end{theorem}

\begin{proof}
The proof is in Appendix B.
\end{proof}
}

{However, the weighted sum rate term in (\ref{P1}) is nonconvex and
is shown to be NP-hard in wireless networks with
interference\cite{ZLuoSzhang}. It is thus extremely difficult, if
not impossible, to get the globally optimal solution to (\ref{P1})
through efficient algorithms in polynomial time. Rather than seeking
global optimality, we will focus on developing low-complexity
algorithms that produce suboptimal solutions to (\ref{P1}).} The
following theorem characterizes the performance of (\ref{P1}) under
suboptimal solutions.

{
\begin{theorem}
Let $\phi$ and $C$ be constants such that $0 < \phi \le 1$ and $C
\ge 0$. Suppose there is an $\epsilon >0 $, such that
\begin{equation}
{\bm{\lambda}} + \epsilon {\bf{1}} \in \phi\mathcal {C}.
\end{equation}

If the suboptimal solution makes (possibly randomized) decisions
every slot to satisfy
\begin{equation}
\sum\limits_{i = 1}^I {{Q_i}(t)\mathbb{E}[{\mu _i}(t)|{\bf{Q}}(t)]}
\ge \phi \left( {\max \sum\limits_{i = 1}^I {{Q_i}(t){\mu _i}} }
\right) - C,\label{scaledproblem}
\end{equation}
then the network is strongly stable.
\end{theorem}
}

\begin{proof}
The proof can follow that for Theorem 6.3 in\cite{neely}.
\end{proof}

{Theorem 3 suggests that the suboptimal solutions that satisfy
(\ref{scaledproblem}) can provide stability whenever the traffic
arrival rates are interior to a $\phi$-scaled version of the
stability region. In this paper, we will develop suboptimal
solutions by relaxing and reformulating the optimization problem in
(\ref{P1}), while it is extremely difficult to quantify the $\phi$
and $C$ that the algorithms can achieve. The stability region
analysis for our proposed algorithms remains challenging and is left
for future work.}

\begin{remark}
For the general case that the arrival rate vector is outside the
stability region $\mathcal {C}$ or the possible reduced one
$\phi\mathcal {C}$, congestion controls are need to constrain the
arrival rate vector into the stability region. In this case, the
problem can be decomposed into a congestion control subproblem and
joint RRH activation and beamforming subproblem by following the
framework in\cite{neelyutility}. By doing so, we can have the
separate congestion control subproblem and joint RRH activation and
beamforming subproblem, and the deterministic worst-case delay bound
can be guaranteed for each traffic queue.
\end{remark}

\section{GSB-based Equivalent Penalized WMMSE Algorithm}
In this section, we will use the GSB approach to solve the optimization problem in
(\ref{P1}).

\subsection{Group Sparse Beamforming Formulation}

Since only a subset RRHs will be active, we can solve the problem by
exploiting the group sparse structure of the aggregated beamforming
vector ${\bf{w}}(t) = [{{\bf{\tilde w}}}_1^T(t),\cdots, {{\bf{\tilde
w}}}_K^T(t)]^T \in {\cal C}^{MIK \times 1}$, where the coefficients
in ${{\bf{\tilde w}}}_k^T(t)$ form a group{\cite{gcpaper}}. When the
RRH $k$ is switched off, all the coefficients in the vector
${\bf{{\tilde w}}}_k$ are 0, which results in the group sparse
structure. The mixed $\ell _1/\ell_p$-norm is shown to be effective
to induce group sparsity and has attracted lots of
attentions\cite{mixednorm}. {In this subsection, we try to construct
a convex relaxation of (\ref{networkpower}), resulting in a weighted
mixed $\ell _1/\ell_2$-norm. Specifically, we first calculate the
\emph{tightest positively homogeneous lower bound} of $p({\bf{w}})$
with the definition ${p_h}({\bf{w}}) = \mathop {\inf }\limits_{\phi
> 0} \frac{{p(\phi {\bf{w}})}}{\phi }, 0 < \phi < \infty$, which is
still nonconvex.  We then calculate the Fenchel conjugate to provide
its convex envelope $\hat p({\bf{w}})$, which is called as the
\emph{tightest convex positively homogeneous lower bound} of
$p({\bf{w}})$ and is given by the following proposition.}

\begin{proposition}
\emph{The tightest convex positively homogeneous lower bound of
(\ref{networkpower}) is given by
\begin{equation} \label{l1_l2}
\hat p ({\bf{w}}(t)) = 2\sum\limits_{k \in \mathcal {K}} {\sqrt
{\frac{{P_k^c}}{{{\eta _k}}}} } ||{{{\bf{\tilde
w}}}_k(t)}|{|_{{2}}},
\end{equation}}
which is a weighted mixed $\ell _1/\ell _2$-norm.
\end{proposition}

\begin{proof}
The proof is in Appendix C.
\end{proof}

The above proposition indicates that the mixed $\ell _1/\ell_2$-norm
can provide a convex relaxation for the cost function
(\ref{networkpower}), thus it can further introduce group sparsity
to $\bf{w}$, that is, many sub-vectors, ${{\bf{\tilde w}}}_k$ will
be 0, which corresponds to inactive RRHs. While the active set of
RRHs corresponds to the non-zero sub-vectors in ${\bf w}$. By
minimizing the weighted mixed $\ell _1/\ell _2$-norm \eqref{l1_l2}
of $\bf{w}$, the zero entries of $\bf{w}$ will be made to align to
the same group ${{\bf{\tilde w}}}_k$, such that the corresponding
RRH is forced to switch off. The weight for each group embraces
additional system parameters. Intuitively, the RRH with a higher
static power consumption and a lower RF power amplifier drain
efficiency will have a high priority being forced to switch off.

Using the weighted mixed $\ell _1/\ell _2$-norm as a surrogate
objective function in (\ref{P1}), we finally have the following
queue-aware group sparse beamforming problem:
\begin{equation}
\begin{array}{l}
\mathop {\max. }\limits_{{\bf{w}}} \sum\limits_{i \in \mathcal {I}}
{{Q_i}{R _i}} - \sum\limits_{k \in \mathcal {K}} {\frac{{2V\sqrt
{P_k^c/{\eta _k}} }}
{{W\tau}}||{{{\bf{\tilde w}}}_k}|{|_2}}, \\
~{\rm{s.t.}}~||{{{\bf{\tilde w}}}_k}||_2^2 \le {P_k}. \\
\end{array}\label{relaxp}
\end{equation}

In the above formulation, the slot index $t$ is skipped to simplify the notation.

\begin{remark}
{The objective function in (\ref{relaxp}) is a convex relaxation to
the original problem (\ref{P1}) using the group sparsity inducing
norm. It has been shown in Proposition 1 that (\ref{relaxp}) is the
tightest convex positively homogeneous lower bound of (\ref{P1}),
that is, among all convex positively homogeneous functions,
(\ref{relaxp}) has the smallest gap with (\ref{P1}). It is very
challenging to quantify the exact performance gap, which normally
requires specific prior information, e.g., in compressive sensing,
the sparse signal is assumed to obey a power law (see Eq. (1.8) in
\cite{ECandes}). However, our problem fundamentally differs from the
existing compressive sensing problems in that we do not have any
prior information about the optimal solution. The optimality
analysis of the queue-aware group sparse beamforming algorithm will
be left to our future work. }

\end{remark}

Next we will design a computationally efficient algorithm that
produces a stationary solution to (\ref{relaxp}) by introducing an
equivalent formulation.

\subsection{Equivalent Formulation and Penalized WMMSE algorithm }

The equivalence between weighted sum rate maximization problem and
WMMSE problem is first established in{\cite{wmmse1}} for
multiple-input and multiple-output (MIMO) broadcast channel and
generalized to MIMO interfering channel in{\cite{wmmse2}}. By
extending the equivalence in{\cite{wmmse1}\cite{wmmse2}}, the
penalized weighted sum rate maximization problem is equivalent to
the following penalized WMMSE problem,
\begin{equation}
\begin{array}{l}
 \mathop {\min.}\limits_{{\bm{\alpha}},{{\bf{u}}},{{\bf{w}}}} ~\sum\limits_{i \in \mathcal {I}} {{Q_i}({\alpha _i}{e_i} - \log {\alpha _i}) } + \sum\limits_{k \in \mathcal {K}} {{\beta _k} ||{{{\bf{\tilde w}}}_k}|{|_{{2}}}}, \\
 ~{\rm{s.t.}}~||{{{\bf{\tilde w}}}_k}||_{{{2}}}^2 \le {P_k}, \\
 \end{array}\label{P2}
\end{equation}
where ${\bm{\alpha}} = \{\alpha _i|i \in \mathcal {I}\}$ is the set
of non-negative mean squared error (MSE) weights, $e_i =
{\bf{u}}_i^H(\sum\limits_{j \in \mathcal {I}}
{{{\bf{H}}_i}{{\bf{w}}_j}{\bf{w}}_j^H{\bf{H}}_i^H + {\sigma
^2}{\bf{I}}} ){{\bf{u}}_i} - 2{\mathop{\rm Re}\nolimits}
\{{\bf{u}}_i^H{{\bf{H}}_i}{{\bf{w}}_i}\} + 1$ is the MSE for
estimating $s_i$, ${\bf{u}} = \{{\bf{u}}_i \in \mathbb{C}^{N \times
1} |i \in \mathcal {I}\}$ is the collection of the receiving vectors
for all UEs,  and ${\beta _k} = \frac{{2V\sqrt {P_k^c/{\eta _k}}
}}{{\log _2^eW\tau}}$ is the parameter that will affect the number
of active RRHs.

{It is worth noting that problem (\ref{P2}) is not jointly convex in
${{\bm{\alpha}},{{\bf{u}}},{{\bf{w}}}}$, while it is convex with
respect to each of the individual optimization variables when fixing
the others. To this end, the block coordinate descent (BCD) method
is utilized to obtain the stationary point of problem (\ref{P2}). As
proven in{\cite{wmmse1}}, once the iterative process converges to a
fixed point of problem (\ref{P2}), the fixed point is also a
stationary point of the problem (\ref{relaxp}). It should be noted
that the stationary point of problem (\ref{relaxp}) or (\ref{P2})
might not be globally optimal.}

Under fixed ${\bf{w}}$ and $\bm{\alpha}$, minimizing the weighted
sum-MSE leads to the well-known MMSE receiver:
\begin{equation}
{\bf{u}}_i = {(\sum\limits_{j \in \mathcal {I}}
{{{\bf{H}}_i}{{\bf{w}}_j}{\bf{w}}_j^H{\bf{H}}_i^H + {\sigma
^2}{\bf{I}}} )^{ - 1}}{{\bf{H}}_i}{{\bf{w}}_i}.\label{ummse}
\end{equation}

With the MMSE receiver, the MSE $e_i$ can be written as
\begin{align}
e_i = 1-{\bf w}_i^H {\bf H}_i^H  \left(\sum\limits_{j
\in \mathcal {I}} {{{\bf{H}}_i}{{\bf{w}}_j}{\bf{w}}_j^H{\bf{H}}_i^H
+ {\sigma ^2}{\bf{I}}} \right)^{-1} {\bf H}_i {\bf w}_i
\end{align}

Under fixed ${\bf{w}}$ and $\bf{u}$, the closed-form $\bm{\alpha}$
can be obtained as follows according to the first-order optimality
conditions:
\begin{equation}
\alpha _i = e_i^{-1}.\label{alphammse}
\end{equation}

Under fixed ${\bf{u}}$ and $\bm{\alpha}$, the optimal ${\bf{w}}$ can
be obtained by solving the following convex problem:
\begin{equation}
\begin{array}{l}
\mathop {\min. }\limits_{{{\bf{w}}}}~\sum\limits_{i \in \mathcal {I}}
{{\bf{w}}_i^H{\bf{C}}} {{\bf{w}}_i} - 2\sum\limits_{i \in \mathcal
{I}} {{\mathop{\rm Re}\nolimits} \{ {{\bf{d}}_i^H}{{\bf{w}}_i}\} } +
\sum\limits_{k \in \mathcal
{K}}\!{{\beta _k}||{{{\bf{\tilde w}}}_k}|{|_{{2}}}},\\
{\rm s.t.}~||{{{\bf{\tilde w}}}_k}||_{_{{2}}}^2 \le {P_k}, \\
\end{array}\label{problemw}
\end{equation}
where ${\bf{C}} = \sum\limits_{j \in \mathcal {I}} {{Q_j}{\alpha
_j}{\bf{H}}_j^H{{\bf{u}}_j}{\bf{u}}_j^H{{\bf{H}}_j}}$ and
${{\bf{d}}_i} = {Q_i}{\alpha _i}{\bf{H}}_i^H{{\bf{u}}_i}$.

The algorithm is summarized in Algorithm \ref{algorithm1}.
\begin{algorithm}[h]
\caption{GSB-based Penalized WMMSE Algorithm}
\begin{algorithmic}[1]
\STATE For each slot $t$, observe the current QSI ${\bf{Q}}(t)$ and
CSI ${\bf{H}}(t)$, then make the queue-aware joint RRH activation
and beamforming according to the following steps: \STATE
\textbf{Initialize} ${\bf{w}}$, ${\bf{u}}$ and $\bm{\alpha}$;
\REPEAT \STATE Fix ${\bf{w}}$, compute the MMSE receiver ${\bf{u}}$
according to (\ref{ummse}) and corresponding MSE $e_i$; \STATE
Update the MSE weight $\bm{\alpha}$ according to (\ref{alphammse});
\STATE Calculate the optimal beamformer ${\bf{w}}$ under fixed
${\bf{u}}$ and $\bm{\alpha}$ by solving (\ref{problemw});\UNTIL
certain stopping criteria is met; \STATE \textbf{Update} the traffic
queue $Q_i(t)$ according to (\ref{qdynamic}).
\end{algorithmic}\label{algorithm1}
\end{algorithm}

%

{The objective function in (\ref{problemw}) contains two parts: the
quadratic part $\sum\limits_{i \in \mathcal {I}}
{{\bf{w}}_i^H{\bf{C}}} {{\bf{w}}_i} - 2\sum\limits_{i \in \mathcal
{I}} {{\mathop{\rm Re}\nolimits} \{ {{\bf{d}}_i^H}{{\bf{w}}_i}\} }$,
and the $\ell _2$-norm part $\sum\limits_{k \in \mathcal
{K}}\!{{\beta _k}||{{{\bf{\tilde w}}}_k}|{|_{{2}}}}$. Unlike the
objective of \emph{standard} group least absolute shrinkage and
selection operator (LASSO) problem, the two parts are functions of
different variables, i.e., ${\bf w}_i$ and ${{\bf{\tilde w}}}_k$,
rather than the same variable. Therefore, existing computationally
efficient algorithms developed for group LASSO in{\cite{admm2}}
cannot be readily applied to solve our \emph{modified} group LASSO
problem (\ref{problemw}).} This fact motivates us to find a new
approach to solve the problem in (\ref{problemw}). Our approach is
based on the famous ADMM algorithm, which will be briefly reviewed
below.

\subsection{Review of ADMM algorithm}

The ADMM algorithm, originally introduced in the 1970s, is a simple
but powerful algorithm that is well suited to distributed convex
optimization, and arbitrary-scale convex optimization. Specifically,
the ADMM is designed to solve the following structured convex
problem{\cite{admm1}}

\begin{equation}
\begin{array}{l}
 \mathop {\min. }\limits_{{\bf{x}} \in {\mathbb{C}^n},{\bf{z}} \in {\mathbb{C}^m}} f({\bf{x}}) + g({\bf{z}}), \\
 ~~~~{\rm{s.t.}}~{\bf{Ax}} + {\bf{Bz}} = {\bf{c}}, \\
~~~~~~~~~{\bf{x}} \in {\mathcal {C}_1}, {\bf{z}} \in {\mathcal {C}_2}, \\
\end{array}\label{ADMMform}
\end{equation}
where ${\bf{A}} \in {\mathbb{C}^{k \times n}}$, ${\bf{B}} \in
{\mathbb{C}^{k \times m}}$, ${\bf{c}} \in {\mathbb{C}^k}$, $f(\cdot)$
and $g(\cdot)$ are convex functions, and $\mathcal {C}_1$ and $\mathcal
{C}_2$ are non-empty convex sets. Correspondingly, the partial
augmented Lagrangian function is given by
\begin{equation}
\begin{array}{l}
{L_\rho }({\bf{x}},{\bf{z}},{\bf{y}}) \!=\! f({\bf{x}}) \!+\!
g({\bf{z}}) \!+\! {\mathop{\rm Re}\nolimits} ({{\bf{y}}^H}({\bf{Ax}}
+ {\bf{Bz}} - {\bf{c}})) \\~~~~~~~~~~~~~+ \frac{\rho }{2}||{\bf{Ax}}
+ {\bf{Bz}} - {\bf{c}}||_2^2,
\end{array}
\end{equation}
where ${\bf{y}} \in {\mathbb{C}^k}$ is the vector of Lagrangian dual variables
associated with the linear equality constraint, and $\rho > 0 $ is a
constant. The ADMM algorithm consists of the following iterations:
\begin{equation}
{{\bf{x}}^{(n+1)}} = \arg \mathop {\min }\limits_{\bf{x}} {L_\rho
}({\bf{x}},{{\bf{z}}^{(n)}},{{\bf{y}}^{(n)}}),
\end{equation}
\begin{equation}
{{\bf{z}}^{(n+1)}} = \arg \mathop {\min }\limits_{\bf{z}} {L_\rho
}({{\bf{x}}^{(n+1)}},{\bf{z}},{{\bf{y}}^{(n)}}),
\end{equation}
\begin{equation}
{{\bf{y}}^{(n+1)}} = {{\bf{y}}^{(n)}} + \rho
({\bf{A}}{{\bf{x}}^{(n+1)}} + {\bf{B}}{{\bf{z}}^{(n+1)}} -
{\bf{c}}).\label{dualupdate}
\end{equation}

It can be easily seen that the ADMM algorithm takes the form of a
decomposition-coordination procedure, in which the solutions to
small local sub-problems are coordinated through dual variable
update (\ref{dualupdate}) to find a solution to a large global
problem. Furthermore, the convergence is established in the
following theorem{\cite{admm1}}.

\begin{theorem}
\emph{Assume that ${{\bf{A}}^T}{\bf{A}} $ and ${{\bf{B}}^T}{\bf{B}}$
are invertible, and the optimal solution to (\ref{ADMMform}) exists.
Then the updated sequence $\{
{{\bf{x}}^{(n)}},{{\bf{z}}^{(n)}},{{\bf{y}}^{(n)}}\}$ is bounded and
the converged $\{ {{\bf{x}}^{(n)}},{{\bf{z}}^{(n)}}\}$ is an optimal
solution of (\ref{ADMMform}).}
\end{theorem}

\subsection{ADMM Algorithm to (\ref{problemw})}

In this subsection, based on splitting the variables ${\bf w}$, the tightly
coupled large-sized problem in \eqref{problemw} will be decomposed into several
sub-problems and solved efficiently by the ADMM algorithm with
closed-form solutions.

To account for the difference between ${\bf w}_i$ and ${\bf {\tilde w}}_k$ in \eqref{problemw},
we introduce a copy ${{\bf{\tilde v}}}_k$ for the original beamformer
${{\bf{\tilde w}}}_k$, and define ${\bf{v}} = [{{\bf{\tilde
v}}}_1^T, ..., {{\bf{\tilde v}}}_K^T]^T \in {\mathbb{C}^{MKI \times
1}}$. The optimization problem in \eqref{problemw} can be equivalently expressed as
\begin{equation}
\begin{array}{l}
\mathop {\min. }\limits_{{{\bf{w}},{\bf{v}}}} \sum\limits_{i \in
\mathcal {I}}{{\bf{w}}_i^H{\bf{C}}} {{\bf{w}}_i}- 2\sum\limits_{i
\in \mathcal {I}}\! {{\mathop{\rm Re}\nolimits} \{
{{\bf{d}}_i^H}{{\bf{w}}_i}\} }
+\sum\limits_{k \in \mathcal {K}} {{\beta _k}||{{{\bf{\tilde v}}}_k}|{|_{{2}}}},\\
{\rm{s.t.}}{~}||{{{\bf{\tilde v}}}_k}||_{_{{2}}}^2 \le {P_k}, \\
{~~~~~}{{\bf{\tilde v}}}_k = {{\bf{\tilde w}}}_k.\\
\end{array}\label{varsplit}
\end{equation}

The partial augmented Lagrangian function of the above problem is given
by
\begin{equation}
\begin{array}{l}
L({\bf{w}},{\bf{v}},{\bf{y}}) = \mathop {\min.
}\limits_{{{\bf{w}},{\bf{v}}}} \!\sum\limits_{i \in \mathcal {I}}\!
{{\bf{w}}_i^H{\bf{C}}} {{\bf{w}}_i} - 2\sum\limits_{i \in \mathcal
{I}} {{\mathop{\rm Re}\nolimits} \{ {{\bf{d}}_i^H}{{\bf{w}}_i}\} }\\
~~~~~~~~~~~~~+\! \sum\limits_{k \in \mathcal {K}} {\!{\beta _k}
||{{{\bf{\tilde v}}}_k}|{|_{{2}}}} \!+\!\! \sum\limits_{k \in
\mathcal {K}} {{\mathop{\!\!\rm Re}\nolimits} \{ {\bf{\tilde
y}}_k^H({{{\bf{\tilde v}}}_k} - {{{\bf{\tilde w}}}_k})\} } \\
~~~~~~~~~~~~~+ \frac{\rho }{2}\sum\limits_{k \in \mathcal {K}} {||{{{\bf{\tilde v}}}_k} - {{{\bf{\tilde w}}}_k}||_2^2},\\
\end{array}\label{P3}
\end{equation}
where ${\bf{y}} = [{\bf{\tilde y}}_1^T, ..., {\bf{\tilde y}}_K^T
]^T$, with ${\bf{\tilde y}}_k =
[{\bf{y}}_{k1}^T,...,{\bf{y}}_{kI}^T]^T \in \mathbb{C}^{MI \times
1}$ being the vector of Lagrangian dual variables for the equality
constraints in (\ref{varsplit}), and $\rho > 0$ is some constant.
The main steps of the ADMM algorithm are summarized in Algorithm
\ref{algorithm2}.
\begin{algorithm}[h]
\caption{ADMM Algorithm for (\ref{problemw})}
\begin{algorithmic}[1]
\STATE \textbf{Initialize} all primal variables ${\bf{w}}^{(0)}$,
${\bf{v}}^{(0)}$ and all dual variables ${\bf{y}}^{(0)}$. \REPEAT
\STATE Solve the following problem and obtain ${\bf{v}}^{(n+1)}$,
\begin{equation}
\begin{array}{l}
\mathop {\min. }\limits_{\bf{v}} L({{\bf{w}}^{(n)}},{\bf{v}},{{\bf{y}}^{(n)}}), \\
~{\rm{s.t.}}~||{{{\bf{\tilde v}}}_k}||_2^2 \le {P_k}; \\
\end{array}\nonumber
\end{equation}
\STATE Solve the following problem and obtain ${\bf{w}}^{(n+1)}$,
\begin{equation}
\mathop {\min. }\limits_{\bf{w}} L({\bf{w}},{{\bf{v}}^{(n +
1)}},{{\bf{y}}^{(n)}});\nonumber
\end{equation}
\STATE Update the multipliers ${\bf{y}}^{(n+1)}$ by
\begin{equation}
{{{\bf{\tilde y}}}_k^{(n + 1)}} = {{{\bf{\tilde y}}}_k^{(n)}} + \rho
({{{\bf{\tilde v}}}_k^{(n+1)}} - {{{\bf{\tilde w}}}_k^{(n+1)}});\nonumber
\end{equation}
\UNTIL certain stopping criteria is met.
\end{algorithmic}\label{algorithm2}
\end{algorithm}

Before obtaining the closed-form expressions for each iteration in
the above algorithm, the convergence of Algorithm 2 is first
discussed.

\begin{theorem}
\emph{Every limit point ${\bf{w}}^{(n)}$ and ${\bf{v}}^{(n)}$
generated by Algorithm 2 is an optimal solution of problem
(\ref{problemw}).}
\end{theorem}

\begin{proof}
The proof is in Appendix D.
\end{proof}

In Algorithm 2, given ${\bf w}$ and ${\bf y}$, the step to obtain $\bf{v}$ can be further
decomposed into $K$ sub-problems, each of which is associated with a
RRH and can be solved in a parallel manner. By completing the
squares in (\ref{P3}), the $k$-th convex sub-problem for RRH $k$ can
be simplified to
\begin{equation}
\begin{array}{l}
\mathop {\min. }\limits_{{{{\bf{\tilde v}}}_k}} \beta _k ||{{{\bf{\tilde v}}}_k}|{|_2} +
\frac{\rho }{2}||{{{\bf{\tilde v}}}_k} - {{{\bf{\tilde w}}}_k} + {{\bf{\tilde y}}_k}/\rho ||_2^2,\\
{\rm{s.t.}}~||{{{\bf{\tilde v}}}_k}||_{_{{2}}}^2 \le {P_k}. \\
\end{array}
\end{equation}

The corresponding Karush-Kuhn-Tucker (KKT) conditions are given by
\begin{equation}
\rho {{\bf{b}}_k} - (\rho + 2{\gamma _k^*}){{\bf{\tilde v}}_k^*} \in
{\beta _k}{\partial (||{{\bf{\tilde v}}_k^*}|{|_2})},
\end{equation}
\begin{equation}
||{{\bf{\tilde v}}_k^*}||_2^2 \le {P_k},{\gamma _k^*} \ge 0,
(||{{\bf{\tilde v}}_k^*}||_2^2 - {P_k}){\gamma _k^*} = 0,
\end{equation}
where ${{\bf{b}}_k} = {{{\bf{\tilde w}}}_k} - {{\bf{\tilde
y}}_k}/\rho $, $\gamma _k^*$ is the optimal Lagrangian multiplier
associated with the power constraint, and ${\partial (||{{\bf{\tilde
v}}_k^*}|{|_2})}$ represents the subdifferential of $\ell _2$-norm
$||\cdot||_2$ at the point ${{\bf{\tilde v}}_k^*}$, which can be
expressed as follows{\cite{subdnorm}}
\begin{equation}
{\partial (||{{\bf{\tilde v}}_k^*}|{|_2})} = \left\{
{\begin{array}{*{20}{c}}
 {\frac{{{{\bf{\tilde v}}_k^*}}}{{||{{\bf{\tilde v}}_k^*}||_2}},~~~~~~~~~~{{\bf{\tilde v}}_k^*} \ne 0} \\
 {\{ {\bf{x}},||{\bf{x}}|{|_2} \le 1\} ,{{\bf{\tilde v}}_k^*} = 0} \\
\end{array}} \right..
\end{equation}

Therefore, we have ${{{\bf{\tilde v}}_k^*} = 0}$ whenever $||\rho
{{\bf{b}}_k}|{|_2} \le {\beta _k}$. When $||\rho {{\bf{b}}_k}|{|_2}
> {\beta _k}$, we have
\begin{equation}
{{\bf{\tilde v}}_k^*} = \frac{{(\rho ||{{\bf{b}}_k}|{|_2} - {\beta
_k}){{\bf{b}}_k}}}{{(\rho + 2{\gamma _k^*})||{{\bf{b}}_k}|{|_2}}}.
\end{equation}

Furthermore, according to the complementary conditions, ${{\gamma
_k^*}} = 0$ if $||\frac{{(\rho ||{{\bf{b}}_k}|{|_2} - {\beta
_k}){{\bf{b}}_k}}}{{\rho ||{{\bf{b}}_k}|{|_2}}}||_2^2 < {P_k}$.
Otherwise, we should have $||\frac{{(\rho ||{{\bf{b}}_k}|{|_2} -
{\beta _k}){{\bf{b}}_k}}}{{(\rho + 2{\gamma
_k^*})||{{\bf{b}}_k}|{|_2}}}||_2^2 = {P_k}$, which indicates that
${\gamma _k^*} = \frac{{\rho ||{{\bf{b}}_k}|{|_2} - {\beta _k} -
\rho \sqrt {{P_k}} }}{{2\sqrt {{P_k}} }}$. In a summary, the
closed-form solution for updating $\bf{v}$ is given by
\begin{equation}
{{\bf{\tilde v}}_k^*} = \left\{ {\begin{array}{*{20}{c}}
 {{\bf{0}},~~~~~~~~~~~~~~~||{{\bf{b}}_k}|{|_2} \le \frac{{{\beta _k}}}{\rho }},~~~~~~~~~~~~~~ \\
 {\frac{{(\rho ||{{\bf{b}}_k}|{|_2} - {\beta _k}){{\bf{b}}_k}}}{{\rho ||{{\bf{b}}_k}|{|_2}}},\frac{{{\beta _k}}}{\rho } < ||{{\bf{b}}_k}|{|_2} < \frac{{{\beta _k}}}{\rho } + \sqrt {{P_k}} }, \\
 {\frac{{{{\bf{b}}_k}\sqrt {{P_k}} }}{{||{{\bf{b}}_k}|{|_2}}},~~~~~~~~~\rm{otherwise}}. ~~~~~~~~~~~~~~~~~ \\
\end{array}} \right.
\end{equation}

As for the step to obtain ${\bf{w}}$, it can be further decomposed
into $I$ sub-problems, each of which is associated with a UE and can
be solved in a parallel manner. By completing the squares in
(\ref{P3}), the $i$-th unconstrained convex sub-problem for UE $i$
can be simplified as
\begin{equation} \label{unconstraintsub}
\mathop {\min. }\limits_{{{\bf{w}}_i}}
{\bf{w}}_i^H{\bf{C}}{{\bf{w}}_i}-2{\mathop{\rm Re}\nolimits}
\{{{\bf{d}}_i^H}{{\bf{w}}_i}\} + \frac{\rho }{2}\sum\limits_{k \in
\mathcal {K}} {||{{{\bf{\tilde w}}}_k} - {{{\bf{\tilde v}}}_k} -
{{\bf{\tilde y}}_k}/\rho ||^2_2}.
\end{equation}

Since the summation term in \eqref{unconstraintsub} is the summation
of squared $\ell_2$-norms, it can be alternatively expressed as
$\sum\limits_{k \in \mathcal {K}} {||{{{\bf{\tilde w}}}_k} -
{{{\bf{\tilde v}}}_k} - {{\bf{\tilde y}}_k}/\rho
||^2_2}=\sum\limits_{i \in \mathcal {I}} {||{{{\bf{ w}}}_i} -
{{{\bf{ v}}}_i} - {{\bf{y}}_i}/\rho ||^2_2}$. Differentiating the
objective with respect to ${\bf{w}}_i$ and setting the result equal
to zero, the optimal ${\bf{w}}_i$ is given by the following
closed-form expression:
\begin{equation}
{\bf{w}}_i^* = {(2{\bf{C}} + \rho {\bf{I}})^{ - 1}}(2{{\bf{d}}_i} +
\rho {{\bf{v}}_i} + {{\bf{y}}_i}),\label{optimalw_i}
\end{equation}
where ${{\bf{y}}_i} = {[{\bf{y}}_{1i}^T,...,{\bf{y}}_{Ki}^T]^T} \in
{\mathbb{C}^{MK \times 1}}$, with ${{\bf{y}}_{ki}} \in
{\mathbb{C}^{M \times 1}} $ being the $i$-th block of ${{\bf{\tilde
y}}_k}$.

\subsection{Discussions on Implementation and Complexity}

{
\subsubsection{Further Reduction of the Number of Active RRHs}
To further decrease the number of active RRHs by enhancing the
group-sparsity for the beamformer, the reweighting procedure that
adaptively reweights the coefficient $\beta_k$ in (\ref{problemw})
can be utilized. Specifically, this can be done in step 6 of
Algorithm 1 by solving problem (\ref{problemw}) and updating the
coefficient $\beta_k$ iteratively, see Section V-C
in\cite{groupgreen} for details. By doing so,the number of active
RRHs is expected to  be smaller than that obtained by solving
problem (\ref{problemw}) only once. However, the enhanced Algorithm
1 with the reweighting procedure will involve two loops, which
results in a high computational complexity. As a result, the
tradeoff between complexity and accuracy should be carefully
considered in the implementation.}

\subsubsection{Parallelized Implementation}

 Each step of
Algorithm 1 and Algorithm 2 can be carried out in a parallel manner.
Except for the MSE weights $\bm{\alpha}$ and the Lagrangian dual
variables $\bf{y}$, the computation for beamformer ${{{\bf{\tilde
v}}}_k}$ and beamformer ${{\bf{w}}_i}$ can be performed in the
parallel computing units of BBU pool for each RRH and each UE,
respectively, without any information exchange. After that, both
$\bm{\alpha}$ and $\bf{y}$ are updated with outputted ${{{\bf{\tilde
v}}}_k}$ and ${{\bf{w}}_i}$. Once the update is done, ${{{\bf{\tilde
v}}}_k}$ and ${{\bf{w}}_i}$ can be calculated simultaneously again
in the parallel computing units of BBU pool.

\subsubsection{Computational Complexity}

The proposed solution is highly efficient as each step of Algorithm
1 and Algorithm 2 is in closed-form. Specifically, the main
computational complexity is related to the matrix inversion in
(\ref{ummse}) and (\ref{optimalw_i}), which have computational
complexity in the order of $\mathcal {O}(N^3)$ and $\mathcal
{O}((MK)^3)$, respectively. Compared with the standard interior
point method, which has a computational complexity in the order of
$\mathcal {O}((MKI)^{3.5})$, the proposed GSB-based solution has a
lower computational complexity, especially for large-scale C-RANs.

\section{RIP-based Equivalent Penalized WMMSE algorithm}
In this section, we will use the RIP approach to solve the optimization problem in (\ref{P1}).

\subsection{Mixed Integer Programming Formulation}

A set of binary variables ${\bf{s}} = \{ {s_k}:{s_k} \in \{ 0,1\} ,k
\in \mathcal {K}\}$ is first introduced to indicate the
active/sleeping states of each RRH, where $s_k = 1$ when RRH $k$ is
activated, and $s_k = 0$ when RRH $k$ is asleep. With ${\bf{s}}(t)$,
the network power consumption model (\ref{networkpower}) can be
rewritten as
\begin{equation}
p({\bf{s}}(t),{\bf{w}}(t)) = \sum\limits_{k \in \mathcal {K}}
{{s_k}(t)(||{{{\bf{\tilde w}}}_k(t)}|{|_2^2} +
P_k^c)}.\label{Eqpowernew}
\end{equation}

{Replacing the penalty term of (\ref{P1}) with the power consumption
model defined in (\ref{Eqpowernew}), we have the following penalized
weighted sum rate maximization problem:}
\begin{equation}
\begin{array}{l}
\mathop {\max. ~}\limits_{{\bf{s}}, {\bf{w}}}\sum\limits_{i \in
\mathcal {I}}
{{Q_i}{R _i}} - {\varphi} \sum\limits_{k \in \mathcal {K}} {{s_k}(||{{{\bf{\tilde w}}}_k}|{|_2^2} + P_k^c)},\\
~{\rm{s.t.}}~||{{{\bf{\tilde w}}}_k}||_2^2 \le P_k,\\
~~~~~s_k = \{0,1\},
\end{array}\label{RIPP1}
\end{equation}
where $\varphi = V/(W\tau \log _2^e)$. As can be seen, the nonlinear
cross-multiplication terms (NCMTs) in the cost function impose a
great challenge on algorithm design. Inspired by the fact that
${{\bf{\tilde w}}}_k$ is equal to $\bf{0}$ if RRH $k$ is turned off,
we can cancel all the NCMTs in the objective function and constraint
as follows:
\begin{equation}
\begin{array}{l}
\mathop {\max. ~}\limits_{{\bf{w}},{\bf{s}}}\sum\limits_{i \in
\mathcal {I}} {{Q_i}{R_i}} - {\varphi}\sum\limits_{k \in \mathcal
{K}} {(||{{{\bf{\tilde
w}}}_k}|{|_2^2} + {s_k}P_k^c)},\\
~{\rm{s.t.}}~||{{{\bf{\tilde w}}}_k}||_2^2 \le {s_k}P_k,\\
~~~~~~s_k = \{0,1\},
\end{array}\label{RIPP2}
\end{equation}

The above problem is a mixed integer nonlinear programming, and some
standard algorithms have been developed to solve it, e.g., the
branch-and-bound (BnB) algorithm{\cite{BnB}}. However, the
computational complexity of BnB is prohibitive for a large-scale
C-RAN. For the worst case, $2^K$ iterations is required, thus their
computational complexity is approximated as $\mathcal
{O}(2^K(KMI)^{3.5})$, which grows exponentially with the number of
RRHs and cannot be applied in practice. Efficient algorithm will be
studied in the following subsection.

\subsection{Equivalent Formulation with Relaxed Integer Programming}

Utilizing the established equivalence between penalized weighted
sum rate maximization and penalized WMMSE in section IV, and
relaxing the binary variable $s_k$ to take continuous value in
$[0,1]$, the problem can be transformed as:
\begin{equation}
\begin{array}{l}
\mathop {\min.
}\limits_{{\bm{\alpha}},{{\bf{u}}},{{\bf{s}}},{{\bf{w}}}}
\sum\limits_{i \in \mathcal {I}} {{Q_i}({\alpha _i}{e_i} - \log
{\alpha _i}) } + {\varphi} \sum\limits_{k \in \mathcal {K}}
{(||{{{\bf{\tilde
w}}}_k}|{|_2^2} + {s_k}P_k^c)}, \\
{\rm{s.t.}}~~ ||{{{\bf{\tilde w}}}_k}||_{{{2}}}^2 \le {s_k}{P_k}, \\
~~~~~~s_k \in [0,1].\\
\end{array}\label{RIPP6}
\end{equation}

The above problem is strongly convex for individual variable
$\bm{\alpha}$ and ${\bf{u}}$ when fixing the rest. Correspondingly,
the unique optimal solutions $\bm{\alpha}^*$ and ${\bf{u}}^*$ are
given by (\ref{ummse}) and (\ref{alphammse}), respectively.
Furthermore, when fixing $\bm{\alpha}$ and ${\bf{u}}$, the
optimization problem for $\bf{s}$ and $\bf{w}$ is convex and is
given by
\begin{equation}
\begin{array}{l}
\mathop {\min. }\limits_{{{\bf{w}},{\bf{s}}}} \! \sum\limits_{i \in
\mathcal {I}} \!{{\bf{w}}_i^H}{\bf{C}}{{\bf{w}}_i} -
2\!\sum\limits_{i \in \mathcal {I}} {{\mathop{\!\rm Re}\nolimits} \{
{{\bf{d}}_i^H}{{\bf{w}}_i}\} } + {\varphi}\!\!\!\sum\limits_{k \in
\mathcal {K}}
{\!\!({s_k}P_k^c \!+\! ||{{{\bf{\tilde w}}}_k}||_2^2)},\\
{\rm{s.t.~}}||{{{\bf{\tilde w}}}_k}||_{_{{2}}}^2 \le {s_k}{P_k}, \\
~~~~~s_k \in [0,1].\\
\end{array}\label{RIPBCD}
\end{equation}

Therefore, a stationary solution can be obtained with the BCD method
by iteratively optimizing over three block variables $\bm{\alpha}$,
${\bf{u}}$ and $\{\bf{s}, {\bf{w}}\}$. The algorithm is summarized
in Algorithm \ref{algorithm3}.

\begin{algorithm}[h]
\caption{RIP-based Penalized WMMSE Algorithm}
\begin{algorithmic}[1]
\STATE For each slot $t$, observe the current QSI ${\bf{Q}}(t)$ and
CSI ${\bf{H}}(t)$, then make the queue-aware joint RRH activation
and beamforming according to the following steps: \STATE
\textbf{Initialize} ${\bf{s}}$, ${\bf{w}}$, ${\bf{u}}$ and
$\bm{\alpha}$; \REPEAT \STATE Fix ${\bf{s}}$ and ${\bf{w}}$, compute
the MMSE receiver ${\bf{u}}$ according to (\ref{ummse}) and the
corresponding MSE. \STATE Update the MSE weight $\bm{\alpha}$
according to (\ref{alphammse}); \STATE Obtain the RRH activation
decision $\bf{s}$ and the beamformer ${\bf{w}}$ under fixed
${\bf{u}}$ and $\bm{\alpha}$ by solving (\ref{RIPBCD});\UNTIL
certain stopping criteria is met; \STATE \textbf{Update} the traffic
queue $Q_i(t)$ according to (\ref{qdynamic}).
\end{algorithmic}\label{algorithm3}
\end{algorithm}

\subsection{RRH Activation and Beamforming based on Lagrangian Dual Decomposition}
The Lagrangian function of problem (\ref{RIPBCD}) is given
by{\cite{boydconvex}}

\begin{equation}
\begin{array}{l}
L({\bf{w}},{\bf{s}},{\bm{\theta }}) = \sum\limits_{i \in \mathcal
{I}} {{\bf{w}}_i^H}{\bf{C}}{{\bf{w}}_i} - 2\sum\limits_{i \in
\mathcal {I}} {{\mathop{\rm Re}\nolimits} \{
{{\bf{d}}_i^H}{{\bf{w}}_i}\} } \\
+ \sum\limits_{k \in \mathcal {K}}
{(\varphi P_k^c - {\theta _k}{P_k}){s_k}} + \sum\limits_{k \in \mathcal {K}} {(\varphi + {\theta _k})||{{{\bf{\tilde w}}}_k}||_2^2}. \\
\end{array}
\end{equation}
where $\bm{\theta} = [\theta _1, \theta _2, ..., \theta _K] \succeq
{\bf{0}}$ is the vector of dual variables associated with the
network power consumption constraints. Correspondingly, the Lagrange
dual function is given by

\begin{equation}
\begin{array}{l}
D(\bm{\theta} ) = \mathop {\min. }\limits_{{\bf{s}},{\bf{w}}} L({\bf{s}},{\bf{w}},\bm{\theta} ), \\
~~~~~~~~~~~{\rm{s.t.}}~{s_k} \in [0,1], \\
\end{array}
\end{equation}
and the dual optimization problem is formulated as
\begin{equation}
\begin{array}{l}
\mathop {\max.~}\limits_{\bm{\theta}} D({\bm{\theta}} ), \\
~{\rm{s.t.}}~{\bm \theta} \succeq {\bf{0}}. \\
\end{array}
\end{equation}

The Lagrangian function $L({\bf{w}},{\bf{s}},{\bm{\theta}})$ is
linear with ${\bm{\theta}}$ for any fixed $\bf{s}$ and ${\bf{w}}$,
while the dual function $D(\bm{\theta})$ is the maximum of these
linear functions. Therefore, the dual optimization problem is always
concave.

Given Lagrangian dual variables, the problem of RRH activation and
beamforming can be decomposed and solved separately. For the problem
of RRH activation, it can be further decomposed into $K$ independent
problems and solved in a parallel manner. The activation problem for
RRH $k$ is given by
\begin{equation}
\begin{array}{l}
\mathop {\min. }\limits_{{s_k}} (\varphi P_k^c - {\theta _k}{P_k}){s_k}, \\
~{\rm s.t.}~0 \le {s_k} \le 1. \\
\end{array}
\end{equation}

It can be easily seen that the optimal solution of $\bf{s}$ is given
by
\begin{equation}
{s_k^*} = \left\{ {\begin{array}{*{20}{c}}
 {0,~~~{\theta _k} \le \varphi P_k^c/{P_k}}, \\
 {1,~~~{\theta _k} > \varphi P_k^c/{P_k}}. \\
\end{array}} \right.\label{RRHonoff}
\end{equation}

{It is worth noting that after the integer relaxation on $s_k(t)$,
we can still obtain optimal solution that is binary. Therefore, the
integer relaxation does not introduce performance loss and there is
no gap with integer relaxation.}

Given Lagrangian dual variables, the problem to derive beamforming
vectors is given by
\begin{equation}
\mathop {\min. }\limits_{{\bf{w}}} \sum\limits_{i \in \mathcal {I}}
{{\bf{w}}_i^H}{\bf{C}} {{\bf{w}}_i} - 2\sum\limits_{i \in \mathcal
{I}} {{\mathop{\rm Re}\nolimits} \{ {\bf{d}}_i^H{{\bf{w}}_i}\} } +
\sum\limits_{k \in \mathcal {K}} {(\varphi +
{\theta_k})||{{{\bf{\tilde w}}}_k}||_2^2},
\end{equation}
which can be further rewritten as
\begin{equation}
\mathop {\min. }\limits_{{\bf{w}}} \sum\limits_{i \in \mathcal {I}}
{{\bf{w}}_i^H}{\bf{C}} {{\bf{w}}_i} - 2\sum\limits_{i \in \mathcal
{I}} {{\mathop{\rm Re}\nolimits} \{ {\bf{d}}_i^H{{\bf{w}}_i}\} } +
\sum\limits_{i \in \mathcal {I}} {{\bf{w}}_i^H{\bm{\Omega}}
{{\bf{w}}_i}},\label{ripwi}
\end{equation}
where ${\bm{\Omega}} = {\rm{diag}}([{\varphi} + {\theta
_1},...,{\varphi} + {\theta _K}] \otimes {{\bf{1}}_M})$ is a
diagonal matrix, $\otimes$ denotes the Kronecker product of two
vectors, and ${{\bf{1}}_M}$ denotes a length $M$ all-one vector. The
problem (\ref{ripwi}) is an unconstrained convex problem, and it can
be further decomposed into $I$ independent problems, each
corresponding to an UE and solved in a parallel manner. Thus,
according to the first-order optimality condition, the optimal
beamforming vector for UE $i$ is given by
\begin{equation}
{{\bf{w}}_i^*} = {({\bf{C}} + {\bm{\Omega}})^{ -
1}}{{\bf{d}}_i}.\label{beamformingRIP}
\end{equation}

\subsection{Lagrangian Dual Variables Update }

As the dual problem is always concave w.r.t. $\bm{\theta}$, we can
adopt the subgradient projection method to solve
it{\cite{boydconvex}}. In particular, it is easy to prove that the
subgradient of the dual function is obtained by
\begin{equation}
\Delta {\theta _k^{(n+1)}} = ||{\bf{\tilde w}}_k^{(n)} ||_2^2 -
s_k^{(n)}{P_k},\label{dualRIP}
\end{equation}
where $s_k^{(n)}$ is the optimal RRH activation according to
(\ref{RRHonoff}) in the $n$-th iteration given $\bm{\theta}^{(n)}$,
and ${\bf{\tilde w}}_k^{(n)}$ can be obtained from the optimal
beamformer for each UE in the $n$-th iteration given
$\bm{\theta}^{(n)}$.

Hence, with the subgradient projection method, the update equation
for the dual variable $\theta _k$ in the $(n+1)$-th iteration is
given by
\begin{equation}
\theta _k^{(n + 1)} = \theta _k^{(n)} + \xi ^{(n+1)}\Delta {\theta
_k^{(n+1)}},
\end{equation}
where $\xi ^{(n+1)}$ is a sufficiently small positive step size. The
whole procedure to solve (\ref{RIPBCD}) is summarized in Algorithm
\ref{algorithm4}.

\begin{algorithm}[h]
\caption{Lagrangian Dual Decomposition Algorithm for (\ref{RIPBCD})}
\begin{algorithmic}[1]
\STATE \textbf{Initialize} $\bf{s}$, $\bf{w}$, $\bm{\theta}$, $n =
0$; \REPEAT \STATE Make the RRH activation decision $\bf{s}$
according to (\ref{RRHonoff}); \STATE Calculate the beamforming
vector ${\bf{w}}_i$ for each UE according to (\ref{beamformingRIP});
\STATE Update the dual variables $\bm{\theta}$ according to
(\ref{dualRIP}); \STATE Set $n = n + 1$; \UNTIL
$||\theta_k^{(n)}-\theta_k^{(n-1)}||< \delta$ or $n > n_{\max}$.
\end{algorithmic}\label{algorithm4}
\end{algorithm}

{\subsection{Discussions on Implementation and Complexity}

\subsubsection{Parallelized Implementation}

Similarly, each step of Algorithm 3 and Algorithm 4 can be carried
out in a parallel manner. Except for the MSE weights $\bm{\alpha}$
and the Lagrangian dual variables ${\bm{\theta}}$, the computation
for RRH activation ${s_k}$ and beamformer ${{\bf{w}}_i}$ can be
performed in the parallel computing units of BBU pool for each RRH
and each UE, respectively, without any information exchange. After
that, both $\bm{\alpha}$ and $\bm{\theta}$ are updated with
outputted ${s_k}$ and ${{\bf{w}}_i}$. Once the update is done,
${s_k}$ and ${{\bf{w}}_i}$ can be calculated simultaneously again in
the parallel computing units of BBU pool.

\subsubsection{Computational Complexity}
The proposed solution is highly efficient as each step of Algorithm
3 and Algorithm 4 is in closed-form. Specifically, the most
computationally intensive operation is the matrix inversion in
(\ref{ummse}) and (\ref{beamformingRIP}), which have computational
complexity in the order of $\mathcal {O}(N^3)$ and $\mathcal
{O}((MK)^3)$, respectively. Compared with the standard interior
point method, which has a computational complexity in the order of
$\mathcal {O}((MKI)^{3.5})$ in our setting, the proposed RIP-based
solution has a lower computational complexity, especially for
large-scale C-RANs. }

\section{Numerical Results}

In this section, we use simulation and numerical results to evaluate and compare
the performance of the proposed algorithms.

\subsection{Scenarios and Parameters Setting}

The pathloss model is $127 + 25\log_{10}(d)$ with $d$ (km) being the
propagation distance. The fast fading is modeled as independent
complex Gaussian random variable distributed according to $\mathcal
{C}\mathcal {N}(0,1)$ and the noise power is -112 dBm. We assumed
that the RRHs are configured with 2 antennas, the UEs are configured
with 1 antennas, and they are uniformly and independently
distributed in the square region [-500 500] $\times$ [-500 500]
meters. Besides, we assume that the mean arrival rate is the same
for all the UEs, i.e., $\lambda_i = \lambda$. We fix the power
budget of each RRH as $P_k = 2$ W. Each point of the the simulation
results is averaged over 4000 slots in the simulations. Our
simulations mainly compare the proposed algorithms with the full
joint processing algorithm (FJP). In FJP algorithm, all the RRHs are
active, and only the transmission power of RRHs is minimized by
solving a penalized weighted sum rate maximization beamforming
design problem based on the Lyapunov optimization. The FJP algorithm
can achieve the highest cooperative beamforming gain with all the
RRHs active, and the results from the FJP algorithm can serve as a
delay performance lower bound of the proposed algorithms.

\subsection{System Performance versus Control Parameter $V$}

We first consider a C-RAN with $K = 9$ RRHs and $I = 6$ UEs. To
indicate the heterogeneous power consumption of different RRHs and
the fronthaul links, we set the static power consumption as $P_k^c
=(2 + k/2)$ W and set the drain efficiency of each RRH as $\eta _k
=0.4$. In Fig. \ref{fig_delay} and \ref{fig_power}, we evaluate the
average queue delay and the average network power consumption
against the control parameter $V$ when the mean arrival rate
$\lambda$ is 1.25 Mbits/slot and 1.75 Mbits/slot, respectively. For
all the algorithms, a larger traffic mean arrival rate always
results in a longer average delay and higher network power
consumptions. This can be explained by the fact that more power is
needed in order to timely transmit larger amount of traffic
arrivals. Under a given mean arrival rate, the average network power
consumption is a monotonically decreasing function in $V$. The rate
of power decreasing starts to diminish with excessive increase of
$V$. On the other hand, a larger $V$ can adversely affect the delay
performance because the average queue length grows linearly with
$V$. This is due to the fact that the system with a larger $V$ will
emphasize less on delay performance but more on the network power
consumption performance. Therefore the parameter $V$ features the
tradeoff between power consumption and delay performance. From both
figures, it is observed that GSB-based WMMSE algorithm always
outperforms RIP-based WMMSE algorithm, but only by a small margin.
Both algorithms achieve significant power saving compared to the FJP
algorithm.

\begin{figure}
\centering
\includegraphics[scale=0.5]{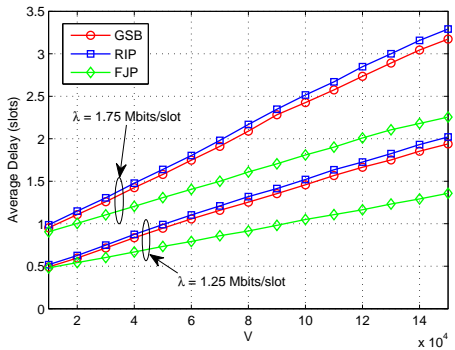}
\caption{Average delay vs. parameter $V$} \label{fig_delay}
\end{figure}

\begin{figure}
\centering
\includegraphics[scale=0.5]{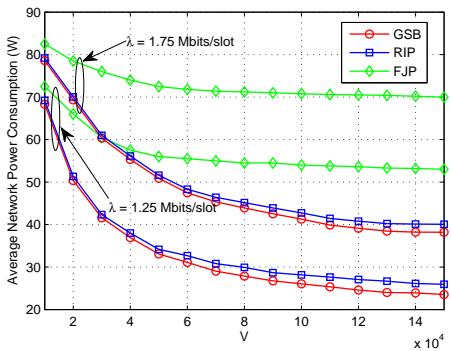}
\caption{Average network power consumption vs. parameter $V$}
\label{fig_power}
\end{figure}

Fig. \ref{fig_RRHs} shows the average number of sleeping RRHs for
the two proposed algorithms against different values of the control
parameter $V$. The average number of sleeping RRHs demonstrates a
similar trend as the average network power consumption performance.
To reduce network power consumption for a large $V$, it is necessary
to turn off as many RRHs and the corresponding fronthaul links as
possible, at the cost of longer average queue length.

\begin{figure}
\centering
\includegraphics[scale=0.5]{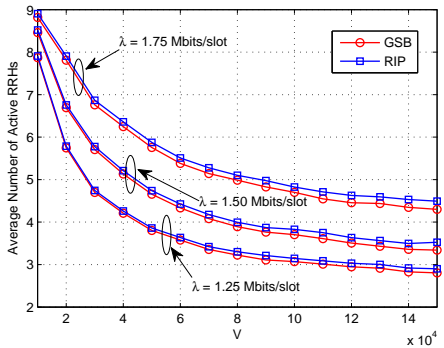}
\caption{Average number of active RRHs vs. parameter
$V$}\label{fig_RRHs}
\end{figure}

\subsection{Quantitative Control of the Power-Delay Tradeoff}

The power-delay tradeoffs of various algorithms are compared in Fig.
\ref{fig_tradeoff}. The different tradeoff points are obtained by
varying the control parameter $V$. The average network power
consumption is decreasing and convex in the average delay. When the
delay is small, slightly increasing the delay requirement can
achieve a significant amount of power saving. When the delay is
excessively large, increasing the delay further results in only a
very small power saving. Fig. \ref{fig_tradeoff} shows that GSB- and
RIP-based algorithm provide significantly better power-delay
tradeoff than the FJP algorithm. In addition, when the average delay
decreases, the power consumption gap between the proposed algorithms
and the FJP algorithm becomes smaller. This means that more RRHs
need to be turned on with a stricter delay requirement. In the
extreme case when all the RRHs are active, all the algorithms will
perform FJP algorithm, yielding the same network power consumption.
In the small delay regime, the proposed algorithms provide a
flexible and efficient means to balance the power-delay tradeoff,
given that a slight loosening of the delay requirement contributes
to significant energy savings. To achieve a certain power-delay
tradeoff, all we need to do is choosing an appropriate control
parameter $V$.

\begin{figure}
\centering
\includegraphics[scale=0.5]{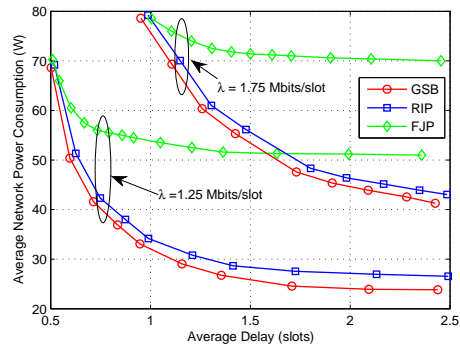}
\caption{Quantitative delay-power tradeoff under different
$\lambda$}\label{fig_tradeoff}
\end{figure}

\subsection{System Performance versus Static Power Consumption}

To compare the average network power consumption with different
static power consumption, we set the static power consumption to be
the same for all RRHs. Fig. \ref{fig_staticpower} shows the average
power consumption as a function of the static power consumption when
the control parameter is $V = 5 \times 10^4$. As expected, the
average network power consumptions for all algorithms are increasing
functions of the static power consumption. In addition, the proposed
algorithms significantly outperform the FJP algorithm, especially in
the high static power consumption regime. When $P^c_k = 4$ W, the
GSB-based algorithm achieves a power saving of 24.6\% and 43.5\%, at
$\lambda = 1.75$ Mbits/slot and 1.25 Mbits/slot, respectively. On
the other hand, the performance gap between the proposed algorithms
and the FJP algorithm decreases as $P^c_k$ decreases. When $P^c_k =
0$, all algorithms require almost the same network power
consumption. When the static power consumption is smaller, more RRHs
will be activated to achieve a higher beamforming gain.

\begin{figure}
\centering
\includegraphics[scale=0.5]{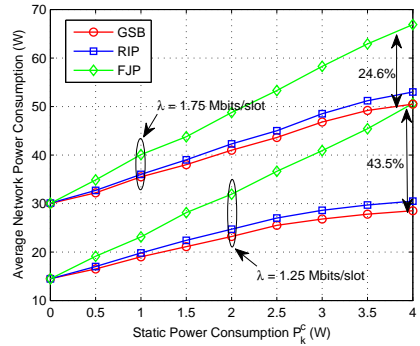}
\caption{Average network power consumption vs. static power
consumption}\label{fig_staticpower}
\end{figure}

{
\subsection{System Performance versus Mean Arrival Rate}

Fig. \ref{fig_rate_delay} compares the average delay with different
mean traffic arrival rates for the two proposed algorithms and the
backward greedy selection (BGS)-based algorithm, which iteratively
switches off one RRH at each step, while re-optimizing the FJP
beamformer for the remaining active RRH set. The BGS-based algorithm
has been shown to often yield optimal or near-optimal solutions for
RRH activation\cite{groupgreen}. We set the control parameter as $V
= 8 \times 10^4$. It is observed that, when the mean traffic arrival
rate is relatively low, the average delays of the two proposed
algorithms are slightly larger than that of the BGS-based algorithm
. Furthermore, the average delays for all algorithms increase
sharply and tend to infinity as the mean arrival rates are beyond
certain thresholds (i.e., the stability regions). Specifically, the
BGS-based algorithm achieves the biggest stability region, followed
by the GSB-based algorithm and the RIP-based algorithm. Therefore,
congestion controls should be adopted to guarantee the queue
stability when the network is with traffic load exceeding current
stability regions. }

\begin{figure}
\centering
\includegraphics[scale=0.5]{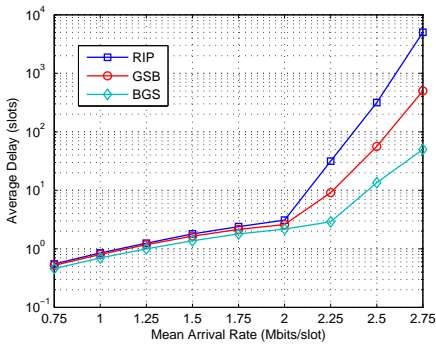}
\caption{Average delay vs. mean arrival rate}\label{fig_rate_delay}
\end{figure}

\subsection{Convergence of the Algorithms} Fig.
\ref{fig_convergence} shows the average number of outer BCD
iterations and the average number of inner iterations required by
the GSB-based WMMSE algorithm and RIP-based WMMSE algorithm,
respectively, with respect to the network scale factor $\Theta$.
Here the network scale factor indicates that the considered C-RAN is
with $K = 9\Theta$ RRHs and $I = 6\Theta$ UEs distributed uniformly
in the squared region $[-500\sqrt \Theta, 500\sqrt \Theta] \times
[-500\sqrt \Theta, 500\sqrt \Theta]$ meters. We can observe that
both GSB-based algorithm and RIP-based algorithm can converge fairly
fast under different network scale, thus both algorithms are highly
scalable to large-scale C-RANs. It is observed that more iterations
are required by the RIP-based solutions compared to that of
GSB-based solutions. Meanwhile, as can be seen in Fig. 5, the
GSB-based solution slightly outperforms the RIP-based solution.
Thus, the GSB-based solution is more preferable in practice.

\begin{figure}
\centering
\includegraphics[scale=0.45]{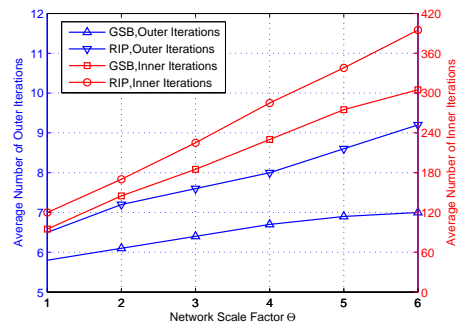}
\caption{Average number of iterations to reach
convergence}\label{fig_convergence}
\end{figure}

\section{Conclusion}

We have developed a joint RRH activation and beamforming algorithm
for a downlink slotted C-RAN, by considering random traffic arrivals
and time-varying channel fadings. The algorithm can achieve flexible
and efficient tradeoff between network power consumption and delay
by adjusting a single parameter. The stochastic optimization problem
of joint RRH activation and beamforming has been transformed into a
penalized weighted sum rate maximization problem based on the
Lyapunov optimization technique. Both GSB and RIP approaches have
been used to reformulate the penalized weighted sum rate
maximization problem. The corresponding algorithms for both
approaches have been proposed, and they were derived based on the
equivalence between the weighted sum rate maximization problem and
the WMMSE problem. The algorithms are guaranteed to converge to a
stationary solution. The solutions do not require any
prior-knowledge of stochastic traffic arrivals and channel
statistics, and can be implemented in a parallel manner. Finally,
the efficiency and the efficacy of the proposed algorithms have been
confirmed by the numerical simulations. For future works, it would
be interesting to consider the queue-aware energy-efficient joint
RRH activation and beamforming algorithms for C-RANs with imperfect
channel state information (CSI) or capacity-limited fronthaul links.

\appendices

\section{Proof of Lemma 1}

By leveraging on the fact that ${(\max [a - b,0] + c)^2} \le {a^2} +
{b^2} + {c^2} - 2a(b - c),\forall a,b,c \ge 0$ and squaring Eq.
(\ref{qdynamic}), we have
\begin{equation}
Q_{i}^2(t + 1) - Q_{i}^2(t) \le \mu _{i}^2(t) + A_{i}^2(t) -
2{Q_{i}}(t)({\mu _{i}}(t) - {A_{i}}(t)),\label{queueinequality}
\end{equation}

{According to the definition of the Lyapunov drift function, we then
have the following expression by summing over all $I$ inequalities
in (\ref{queueinequality}) and taking expectation over both sides,}
\begin{equation}
\begin{array}{l}
\Delta ({\bf{Q}}(t)) \le \frac{1}{2}\mathbb{E}\left[\sum\limits_{i
\in \mathcal {I}} {\mu _i^2(t) + A_i^2(t)|{\bf{Q}}(t)} \right]
\\~~~~~~~~~~- \sum\limits_{i \in \mathcal {I}} {{Q_i}(t)\mathbb{E}[{\mu
_i}(t) - {A_i}(t)|{\bf{Q}}(t)]}. \end{array}\label{driftbound}
\end{equation}

Let $B \ge \frac{1}{2}\sum\limits_{i = 1}^I {\mathbb{E}[{
{A_i^2}(t)} + {{\mu_i^2}(t)}|{\bf{Q}}(t)]}$. Finally, the
upper bound in
(\ref{driftpenaltybound}) can be obtained by adding
$V\mathbb{E}[p(\mathcal {A}(t), {\bf{w}}(t))|{\bf{Q}}(t)]$ to both sides of
(\ref{driftbound}).

\section{Proof of Theorem 2}

Suppose that the traffic arrivals with mean arrival rate
${\bm{\lambda}} = ({\lambda _1},...,{\lambda _I})$ is strictly
interior to the stability region $\mathcal {C}$ (Definition 2) such
that ${\bm{\lambda}} + \epsilon{\bf{1}} \in \mathcal {C}, \forall
\epsilon >0$. Since channel conditions are i.i.d. over slots,
according to Theorem 4.5 in\cite{neely}, there exists a stationary
randomized control policy that is independent of ${\bf{Q}}(t)$ and
yields
\begin{equation}
\begin{array}{l}
\mathbb{E}[{\mu _i}(t)|{\bf{Q}}(t)] = \mathbb{E}[{\mu _i}(t)] \ge {\lambda _i} + \epsilon ,\forall i, \\
\mathbb{E}[p(\mathcal {A}(t),{\bf{w}}(t))|{\bf{Q}}(t)] = \mathbb{E}[p(\mathcal {A}(t),{\bf{w}}(t))] = \bar p(\epsilon ). \\
\end{array}\label{ranpolicy}
\end{equation}

As the stationary randomized control policy is simply a particular
control policy, it certainly satisfies (\ref{driftpenaltybound}) in
Lemma 1. In addition, since (\ref{P1}) is obtained by minimizing the
right-hand-side (R.H.S.) of (\ref{driftpenaltybound}) among all
feasible policies (including the stationary randomized control
policy), by combining (\ref{ranpolicy}) with
(\ref{driftpenaltybound}), we have

\begin{equation}
\begin{array}{l}
\Delta ({\bf{Q}}(t)) \!+\!\! V\mathbb{E}[p(\mathcal {A}(t),
{\bf{w}}(t))|{\bf{Q}}(t)] \le B +
\\V\mathbb{E}[ p(\mathcal {A}(t),
{\bf{w}}(t))|{\bf{Q}}(t)] + \sum\limits_{i \in \mathcal {I}}
{Q_i}(t)\mathbb{E}[{A_i}(t) - \mu _i(t)|{\bf{Q}}(t)]\\
\le B + V\bar p(\epsilon ) - \varepsilon \sum\limits_{i \in \mathcal
{I}} {{Q_i}(t)}.
\end{array}
\end{equation}

Using the results in the proof of Theorem 1, it follows that
$\sum\limits_{i \in \mathcal {I}} {{Q_i}(t)}  \le \frac{{B + V\bar
p(\epsilon )}}{\epsilon }$, which proves that solving (\ref{P1})
optimally stabiles all the queues.

\section{Proof of Proposition 1}

Let $L=KMI$ and define the index set $\mathcal {V} = \{1, 2, ...,
L\}$, then we have ${\bf{w}} = [w_l:l \in \mathcal {V}]$. Define the
set $\mathcal {I} _k = \{(k-1)MI + 1, ..., kMI\}$ as a partion of
$\mathcal {V}$, then we have ${{\bf{\tilde w}}_k} =
{{\bf{w}}_{{\mathcal {I}_k}}} = [w_l:l \in \mathcal {I} _k]$.
Furthermore, define the support of beamformer $\bf{w}$ as $\mathcal
{T}({\bf{w}}) = \{l|w_l \neq 0\}$, then the power consumption model
can be rewritten as
\begin{equation}
p({\bf{w}}) = \sum\limits_{k \in \mathcal {K}} {(\frac{1}{{{\eta
_k}}}||{{\bf{w}}_{{\mathcal {I}_k}}}||_2^2 + P_k^c}
\mathbb{I}(\mathcal {T}({\bf w}) \cap {\mathcal {I}_k} \ne \emptyset
)),
\end{equation}
where $\mathbb{I}(\mathcal {E})$ is an indicator function with value
1 if the event ${\mathcal{E}}$ is true and 0 otherwise. To simplify
notation, let $T({\bf{w}}) = \sum\limits_{k \in \mathcal {K}}
{\frac{1}{\eta _k}||{{\bf{w}}_{{\mathcal {I}_k}}}||_2^2} $, and let
$F(\mathcal {T}({\bf w})) = P_k^c\mathbb{I}(\mathcal {T}({\bf w})
\cap {\mathcal {I}_k} \ne \emptyset )$. It can been seen that
$F(\cdot)$ is combinatorial in ${\bf w}$ and is non-convex. We will
obtain a convex relaxation of the combinatorial objective function.

{We first construct the tightest positively homogeneous lower bound
of $p({\bf{w}})$,} which is given by{\cite{convexrelax}}
\begin{equation}
\begin{array}{l}
{p_h}({\bf{w}}) = \mathop {\inf }\limits_{\phi > 0} \frac{{p(\phi
{\bf{w}})}}{\phi } = \mathop {\inf }\limits_{\phi > 0} \phi
T({\bf{w}}) + \frac{1}{\phi }F(\mathcal {T}({\bf{w}}))\\
~~~~~~~~= 2\sqrt {F(\mathcal {T}({\bf{w}}))T({\bf{w}})}.
\end{array}\label{lowerbound}
\end{equation}
The last equality in (\ref{lowerbound}) is obtained by solving
$\partial \frac{{p(\phi {\bf{w}})}}{\phi }/\partial \phi  = 0$.
{However, ${p_h}({\bf{w}})$ is still non-convex in ${\bf w}$}.

{We next calculate the convex envelope of ${p_h}({\bf{w}})$}. Define
diagonal matrices ${\bf{\Lambda}} \in \mathbb{R}^{L \times L}$,
${\bf{\Xi}} \in \mathbb{R}^{L \times L}$ with the $k$-th diagonal
block being $\eta _k{\bf{I}}_{MI}$ and ${\frac{1}{\eta
_k}}{\bf{I}}_{MI}$, respectively. The Fenchel conjugate of
${p_h}({\bf{w}})$ is given by
\begin{equation}
\begin{array}{l}
p_h^*({\bf{z}})\! =\!\! \mathop {\sup }\limits_{{\bf{w}} \in
{\mathbb{C}^L}} ({{\bf{z}}^T}{{\bf{\Lambda}}^T}{\bf{\Xi}}{\bf{w}}
\!-\! 2\sqrt {F(\mathcal {T}({\bf{w}}))T({\bf{w}})} ) \\
= \mathop{\sup }\limits_{\mathcal {X} \in \mathcal {V}} \mathop
{\sup }\limits_{{{\bf{w}}_\mathcal {X}} \in {\mathbb{C}^{|\mathcal
{X}|}}} ({\bf{z}}_\mathcal {X}^T{\bf{\Lambda}}_{\mathcal
{X}}^T{{\bf{\Xi}}_{\mathcal {X}}}{{\bf{w}}_\mathcal {X}} -
2\sqrt {F(\mathcal {X})T({{\bf{w}}_\mathcal {X}})} )\\
= \left\{ {\begin{array}{*{20}{c}} {0,{\rm{~~~~if~}}{\hat p
^*}({\bf{z}}) = \mathop {\sup }\limits_{\mathcal {X} \subseteq
\mathcal {V},\mathcal {X} \ne \emptyset }
\frac{{||{{\bf{z}}_\mathcal {X}}{{\bf{\Lambda}}_{\mathcal
{X}}}|{|_2}}}{{2\sqrt {F(\mathcal {X})} }} \le 1}, \\
{\infty, {\rm{~~~~otherwise.~~~~~~~~~~~~~~~~~~~~~~~~~~}}} \\
\end{array}} \right.
\end{array}\label{Fen}
\end{equation}
where ${\bf{z}}_\mathcal {X}$ is the $|\mathcal {X}|$-dimensional
vector formed with the entries of ${\bf z}$ indexed by $\mathcal
{X}$ (similarly for $\bf{w}$), and ${\bf{\Lambda}}_{\mathcal {X}}$
is the $|\mathcal {X}| \times |\mathcal {X}|$-dimensional matrix
formed with both the rows and the columns of ${\bf{\Lambda}}$
indexed by $\mathcal {X}$ (similarly for ${\bf{\Xi}}$). Consider the
norm $\hat p({\bf{w}})$ whose dual norm is defined as $\hat
p^*({\bf{z}})$ in Eq. (\ref{Fen}). {According to Proposition 2
in{\cite{convexrelax}}, $\hat p({\bf{w}})$ is the convex envelope of
$p({\bf{w}})$}.

Therefore, the tightest convex positively homogenous lower bound of
$p(\bf{w})$ {has the following inequality:}

\begin{equation}
\begin{array}{l}
\hat p({\bf{w}}) = \mathop {\sup }\limits_{{\hat p^*}({\bf{z}}) \le 1} {{\bf{w}}^T}{\bf{z}} \le \mathop {\sup }\limits_{{\hat p^*}({\bf{z}}) \le 1} \sum\limits_{k \in \mathcal {K}} {||{{\bf{w}}_{{\mathcal {I}_k}}}|{|_2}||{{\bf z}_{{\mathcal {I}_k}}}|{|_2}} \\
\le \mathop {\sup }\limits_{{\hat p^*}({\bf{z}}) \le 1} \left( {\sum\limits_{k \in \mathcal {K}} {\sqrt {\frac{{P_k^c}}{{{\eta _k}}}} ||{{\bf{w}}_{{\mathcal {I}_k}}}|{|_2}} } \right)\left( {\mathop {\max }\limits_{k \in \mathcal {K}} \sqrt {\frac{{{\eta _k}}}{{P_k^c}}} ||{{\bf{z}}_{{\mathcal {I}_k}}}|{|_2}} \right)\\
= 2\sum\limits_{k \in \mathcal {K}} {\sqrt {\frac{{P_k^c}}{{{\eta _k}}}} ||{{\bf{w}}_{{\mathcal {I}_k}}}|{|_2}}, \\
\end{array}
\end{equation}
{which is obtained using the norm properties. Actually, the above
inequality always holds with equality.} Specifically, let
${{{\bf{\bar z}}}_{{\mathcal {I}_k}}} = 2\sqrt
{\frac{{P_k^c}}{{{\eta _k}}}} \frac{{{\bf{w}}_{{\mathcal
{I}_k}}^H}}{{||{\bf{w}}_{{\mathcal {I}_k}}^H|{|_2}}}$ such that
${\hat p^*}({\bf{\bar z}}) = 1$, then we have
\begin{equation}
\hat p({\bf{w}}) = \mathop {\sup }\limits_{{\hat p^*}({\bf{z}}) \le
1} {{\bf{w}}^T}{\bf{z}} \ge \sum\limits_{k \in \mathcal {K}}
{{\bf{w}}_{{\mathcal {I}_k}}^T{{{\bf{\bar z}}}_{{\mathcal {I}_k}}}}
= 2\sum\limits_{k \in \mathcal {K}} {\sqrt {\frac{{P_k^c}}{{{\eta
_k}}}} ||{{\bf{w}}_{{\mathcal {I}_k}}}|{|_2}},
\end{equation}
{which is obtained using the definition of convex envelope.
Therefore, we finally have $2\sum\limits_{k \in \mathcal {K}} {\sqrt
{\frac{{P_k^c}}{{{\eta _k}}}} ||{{\bf{w}}_{{\mathcal {I}_k}}}|{|_2}}
\le \hat p({\bf{w}}) \le 2\sum\limits_{k \in \mathcal {K}} {\sqrt
{\frac{{P_k^c}}{{{\eta _k}}}} ||{{\bf{w}}_{{\mathcal
{I}_k}}}|{|_2}}$, i.e., (\ref{l1_l2}).}

\section{Proof of Theorem 6}

By comparing problem (\ref{ADMMform}) and problem (\ref{varsplit}),
when ${\bf{x}} = {\bf{v}}$ and ${\bf{z}} = {\bf{w}}$, we can observe
that
\begin{equation}
\begin{array}{l}
f({\bf{x}}) = \sum\limits_{k \in \mathcal {K}} {{\beta
_k}||{{{\bf{\tilde v}}}_k}|{|_{{2}}}},\\
g({\bf{z}}) = \sum\limits_{i \in \mathcal {I}}{{\bf{w}}_i^H{\bf{C}}}
{{\bf{w}}_i}- 2\sum\limits_{i \in \mathcal {I}}\! {{\mathop{\rm
Re}\nolimits} \{ {{\bf{d}}_i^H}{{\bf{w}}_i}\} },\\
{\bf{A}} = {\bf{I}}, {\bf{B}} = {\bf{I}}, {\bf{c}} = {\bf{0}},\\
\mathcal {C}_1 = \left( {{\bf{x}}|~||{{{\bf{\tilde
v}}}_k}||_{_{{2}}}^2 \le {P_k}, \forall k \in \mathcal {K}} \right),
\mathcal {C}_2 = {{\bf{z}}}.
\end{array}
\end{equation}

Since ${{\bf{A}}^T}{\bf{A}} = {\bf{I}}$ and ${{\bf{B}}^T}{\bf{B}} =
{\bf{I}}$ are invertible, and both $\mathcal {C}_1 $ and $\mathcal
{C}_2$ are convex sets, then according to Theorem 3, we can conclude
that every limit point ${\bf{w}}^{(n)}$ and ${\bf{v}}^{(n)}$
generated by Algorithm 2 is an optimal solution of problem
(\ref{problemw}).


\begin{thebibliography}{99}

\bibitem{chiling}
C. I, C. Rowell, S. Han, Z. Xu, G. Li, and Z. Pan, ``Toward green
and soft: a 5G perspective,'' \emph{IEEE Commun. Mag.}, vol. 52, no.
2, pp. 66-73, Feb. 2014.



\bibitem{cmcc}
China Mobile Research Institute, ``C-RAN: The road towards green
RAN,'' \emph{White Paper}, ver. 3.0, Dec. 2013.


\bibitem{5G1}
P. Rost, C. J. Bernardos, A. D. Domenico, M. D. Girolamo, M. Lalam,
A. Maeder, D. Sabella, and D. W${\rm{\ddot u}}$bben, ``Cloud
technologies for flexible 5G radio access networks,'' \emph{IEEE
Commun. Mag.}, vol. 52, no. 5, pp. 68-76, May 2014.


\bibitem{CranMGMN}
NGMN alliance, ``Suggestion on potential solution to C-RAN,'' Jan.
2013.

\bibitem{5G}
C. Wang, F. Haider, X. Gao, X. You, Y. Yang, D. Yuan, H. M. Aggoune,
H. Haas, S. Fletcher, and E. Hepsaydir, ``Cellular architecture and
key technologies for 5G wireless communication networks,''
\emph{IEEE Commun. Mag.}, vol. 52, no. 2, pp. 122-130, Feb. 2014.



{



\bibitem{quek}
J. Zhao, T. Q. S. Quek, and Z. Lei, ``Coordianted multipoint
transmission with limited backhaul data transfer,'' \emph{IEEE Trans
Wireless Commun.}, vol. 12, no. 6, pp. 2762-2774, Jun. 2013.

\bibitem{lau}
F. Zhuang and V. K. N. Lau, ``Backhaul limited asymmetric
cooperation for MIMO cellular networks via semidefinite
relaxation,'' \emph{IEEE Trans. Signal Process.}, vol. 62, no. 3,
pp. 684-693, Feb. 2014.

\bibitem{VNHa}
V. N. Ha, L. B. Le, and N. D. Dao, ``Coordinated multipoint (CoMP)
transmission design for Cloud-RANs with limited fronthaul capacity
constraints,'' \emph{IEEE Trans. Veh. Tech.}, early access, 2015.


\bibitem{BDai}
B. Dai and W. Yu, ``Sparse beamforming and user-centric clustering
for downlink cloud radio access network,'' \emph{IEEE Access}, vol.
2, pp. 1326-1339, Oct. 2014.


\bibitem{MINP2}
Y. Cheng, M. Pesavento, and A. Philipp, ``Joint network optimization
and downlink beamforming for CoMP transmissions using mixed integer
conic programming,'' \emph{IEEE Trans. Signal Process.}, vol. 61,
pp. 3972-3987, Aug. 2013.


\bibitem{SLuo}
S. Luo, R. Zhang, and T. J. Lim, ``Downlink and uplink energy
minimization through user association and beamforming in cloud
RAN,'' \emph{IEEE Trans. Wireless Commun.}, vol. 14, no. 1, pp.
494-508, Jan. 2015.




\bibitem{fronthaulpower}
S. Tombaz, P. Monti, K. Wang, A. V${\rm{\ddot a}}$stberg, M.
Forzati, and J. Zander, ``Impact of backhauling power consumption on
the deployment of heterogeneous mobile networks,'' in \emph{Proc.
IEEE Global Commun. Conf. (GLOBECOM)}, Houston, TX, USA, Dec. 2011,
pp. 1-5.

\bibitem{groupgreen}
Y. Shi, J. Zhang, and K. B. Letaief, ``Group sparse beamforming for
green cloud-RAN,'' \emph{IEEE Trans Wireless Commun.}, vol. 13, no.
5, pp. 2809-2823, May 2014.

\bibitem{BDaiJSAC}
B. Dai and W Yu, ``Energy efficiency of downlink transmission
strategies for cloud radio access networks,'' \emph{IEEE J. Sel.
Areas Commn.}, Availabe: http://arxiv.org/abs/1601.01070v2.


\bibitem{cuiying}
Y. Cui, V. K. N. Lau, R. Wang, H. Huang, and S. Zhang, ``A survey on
delay-aware resource control for wireless systems - large deviation
theory, stochastic Lyapunov drift and distributed stochastic
learning,'' \emph{IEEE Trans. Inf. Theory}, vol. 58, no. 3, pp.
1677-1701, Mar. 2012.

\bibitem{MNeelyFair}
M. J. Neely, E. Modiano, and C. Li, ``Fairness and optimal
stochastic control for heterogeneous networks,'' \emph{IEEE/ACM
Trans. Netw.}, vol. 16, no. 2, pp. 396-409, Apr. 2008.

\bibitem{LongBLe}
L. B. Le, E. Modiano, and N. B. Shroff, ``Optimal control of
wireless networks with finite buffers,'' \emph{IEEE/ACM Trans.
Netw.}, vol. 20, no. 4, pp. 1316-1329, Aug. 2012.


\bibitem{NeelyDelayAnalysis}
M. J. Neely, ``Delay analysis for maximal scheduling with flow
control in wireless networks with bursty traffic,'' \emph{IEEE/ACM
Trans. Netw.}, vol. 17, no. 4, pp. 1146-1159, Aug. 2009.





\bibitem{r1}
M. J. Neely, ``Energy optimal control for time-varying wireless
networks,'' \emph{IEEE Trans. Inf. Theory}, vol. 52. no. 7, pp.
2915-2934, Jul. 2006.


\bibitem{r2}


H. Ju, B. Liang, J. Li, and X. Yang, ``Dynamic power allocation for
throughput utility maximization in interference-limited networks,''
\emph{IEEE Wireless Commun. Lett.}, vol. 2, no. 1, pp. 22-25, Jan.
2013.


\bibitem{park}
S. Park, O. Simeone, O. Sahin, and S. Shamai, ``Joint precoding and
multivariate backhaul compression for the downlink of cloud radio
access networks,'' \emph{IEEE Trans. Signal Process.}, vol. 61, no.
22, pp. 5646-5658, Nov. 2013.
}




\bibitem{powermodel}
G. Auer, V. Giannini, C. Desset, I.  Godor, P. Skillermark, M.
Olsson, M. A. Imran, D. Sabella, M. J. Gonzalez, O. Blume, and A.
Fehske, ``How much energy is needed to run a wireless network?''
\emph{IEEE Wireless Commun.}, vol. 18, pp. 40-49, Oct. 2011.

\bibitem{PON}
Y. A. Sambo, M. Z. Shakir, K. A. Qaraqe, E. Serpedin, M. A. Imran,
``Expanding cellular coverage via cell-edge deployment in
heterogeneous networks: spectral efficiency and backhaul power
consumption perspectives,'' \emph{IEEE Commun. Mag.}, vol. 52, no.
6, pp. 140-149, Jun. 2014.





\bibitem{neely}
M. J. Neely, \emph{Stochastic Network Optimization with Application
to Communication and Queueing Systems}, Morgan\&Claypool., 2010.


\bibitem{neelyutility}
M. J. Neely, ``Delay-based network uility maximization,''
\emph{IEEE/ACM Trans. Netw.}, vol. 21. no. 1, pp. 41-54, Feb. 2013.


\bibitem{ZLuoSzhang}
Z. Q. Luo and S. Zhang, ``Dynamic spectrum management: Complexity
and duality,'' \emph{IEEE J. Sel. Areas Signal Process.}, vol. 2,
no. 1, pp. 57-73, Feb. 2008.



\bibitem{gcpaper}
J. Li, J. Wu, M. Peng, W. Wang, and V. K. N. Lau, ``Queue-aware
joint remote radio head activation and beamforming for green cloud
radio access networks,'' in \emph{Proc. IEEE Global Commun. Conf.
(GLOBECOM)}, San Diego, CA, USA, Dec. 2015.

\bibitem{mixednorm}
Y. C. Eldar and G. Kutyniok, \emph{Compressed Sensing: Theory and
Applications}, Cambridge, U.K.: Cambridge Univ. Press, 2012.

\bibitem{ECandes}
E. Candes and T. Tao, ``Near-optimal signal recovery from random
projections: Universal encoding strategies?'' \emph{IEEE Trans. Inf.
Theory}, vol. 52, no. 12, pp. 5406-5425, Dec. 2006.


\bibitem{wmmse1}
S. S. Christensen, R. Agarwal, E. Carvalho, and J. M. Cioffi,
``Weighted sum-rate maximization using weighted MMSE for MIMO-BC
beamforming design,'' \emph{IEEE Trans Wireless Commun.}, vol. 7,
no. 12, pp. 4792-4799, Dec. 2008.

\bibitem{wmmse2}
Q. Shi, M. Razaviyayn, Z. Luo, and C. He, ``An iteratively weighted
MMSE approach to distributed sum-utility maximization for a MIMO
interfering broadcast channel,'' \emph{IEEE Trans. Signal Process.},
vol. 59, no. 9, pp. 4331-4340, Sep. 2011.




\bibitem{admm2}
W. Deng, W. Yin, and Y. Zhang, ``Group sparse optimization by
alternating direction method,'' Technical Report, Rice University,
2011.

\bibitem{admm1}
S. Boyd, N. Parikh, E. Chu, B. Peleato, and J. Eckstein,
``Distributed optimization and statistical learning via the
alternating direction method of multipliers,'' \emph{Found. Trends
Mach. Learn.}, vol. 3, no. 1, pp. 1-122, 2011.


\bibitem{subdnorm}
F. Bach, R. Jenatton, J. Mairal, and G. Obozinski, ``Convex
optimization with sparsity-inducing norms,'' in \emph{Optimization
for Machine Learning}, S. Sra, S. Nowozin, and S. J. Wright, Eds.
Cambridge, MA: MIT Press, 2011.



\bibitem{BnB}
D. Wei and A. V. Oppenheim, ``A branch-and-bound algorithm for
quadratically-constrained sparse filter design,'' \emph{IEEE Trans.
Signal Process.}, vol. 61, no. 4, pp. 1006-1018, Feb. 2013.

\bibitem{boydconvex}
S. Boyd and L. Vandenberghe, \emph{Convex Optimization}, Cambridge,
U.K.: Cambridge University Press, 2004.


\bibitem{convexrelax}
G. Obozinski and F. Bach, ``Convex relaxation for combinatorial
penalties,'' a Technical Report 00694765, HAL, 2012.


\bibitem{bcd}
P. Teseng, ``Convergence of a block coordinate descent method for
nondifferentiable minimization,'' \emph{J. Opt. Theory and App.},
vol. 109, no. 3, pp. 475-494, Jun. 2001.

\end{thebibliography}
\end{document}